\documentclass[12pt]{amsart}

\usepackage{amsmath,amssymb}
\usepackage[utf8]{inputenc}
\usepackage{graphicx}
\usepackage{enumitem}
 \newif\ifanswers
  \answerstrue

 \ifanswers
\usepackage{pst-eucl,pstricks-add}
\usepackage{pst-all}
\usepackage{pdftricks}

 \fi
\usepackage[
        colorlinks,
]{hyperref}

\newtheorem{thm}{Theorem}
\newtheorem{lem}[thm]{Lemma}
\newtheorem{prop}[thm]{Proposition}
\newtheorem{cor}[thm]{Corollary}
\newtheorem{rem}[thm]{Remark}
\newtheorem*{claim}{Claim}

\newcommand{\Z}{{\mathbb Z}} 
\newcommand{\Q}{{\mathbb Q}} 
 
\newcommand{\R}{{\mathbb R}} 
 
\newcommand{\Op}{\operatorname{Op}}

\newcommand{\noXi}{{\Xi}} 

\title[Scars for point scatterers]{Scarred eigenstates for arithmetic
  toral point scatterers} 

\begin{document}
\date{August 12, 2015}
\author{P\"ar Kurlberg}
\address{Department of Mathematics, KTH Royal Institute of Technology,
  SE-100 44
Stockholm, Sweden; e-mail: kurlberg@math.kth.se}
\author{Lior Rosenzweig}
\address{Department of Mathematics, KTH Royal Institute of Technology,
  SE-100 44
Stockholm, Sweden; e-mail: liorr@math.kth.se}

\begin{abstract}
  We investigate eigenfunctions of the Laplacian perturbed by a delta
  potential on the standard tori $\R^d/2 \pi\Z^d$ in dimensions
  $d=2,3$.
Despite {\em quantum ergodicity}  holding for
  the set of ``new'' eigenfunctions 
we show that there is {\em scarring} in the momentum
representation for $d=2,3$, as well as in the position representation
for $d=2$ (i.e., the  eigenfunctions fail to equidistribute in phase
space along an infinite subsequence of new eigenvalues.)
For $d=3$, scarred eigenstates are quite rare, but for $d=2$  scarring in the momentum representation is very common --- with
$N_{2}(x) \sim x/\sqrt{\log x}$ denoting the counting function for the
new eigenvalues below $x$, there are $\gg N_{2}(x)/\log^A x$
eigenvalues corresponding to momentum scarred eigenfunctions.

\end{abstract}
\thanks{P.K. and L.R.  were partially supported by grants from the
  G\"oran Gustafsson Foundation for Research in Natural Sciences and
  Medicine, and the Swedish Research Council (621-2011-5498).} 
\maketitle



 
\section{Introduction}

\label{sec:introduction}

A basic question in Quantum Chaos is the classification of quantum
limits of energy eigenstates of quantized Hamiltonians.  For example,
if the classical dynamics is given by the geodesic flow on a compact
Riemannian manifold $M$, the quantized Hamiltonian is given by the
positive Laplacian $-\Delta$ acting on $L^{2}(M)$.  With
$\{\psi_{\lambda}\}_{\lambda}$ denoting Laplace eigenfunctions giving
an orthonormal basis for $L^{2}(M)$, a quantum limit is a weak$^{*}$
limit of $|\psi_{\lambda}(x)|^{2}$ along any subsequence of
eigenvalues $\lambda$ tending to infinity.  More generally, given a
smooth observable, i.e.  a smooth function $f$ on the unit cotangent
bundle $S^{*}(M)$, its quantization is defined as a 
pseudo-differential operator $\Op(f)$, and one wishes to understand
possible limits of the distributions
$$
f \to \langle \Op(f) \psi_{\lambda}, \psi_{\lambda} \rangle
$$
on $C^{\infty}( S^{*}(M))$, as $\lambda \to \infty$.  If $M$ 
has negative curvature (``strong chaos''), the celebrated
Quantum Unique Ergodicity (QUE) conjecture by Rudnick and Sarnak
\cite{RS94} 
asserts that the only quantum limit is given by the uniform, or
Liouville, measure on $S^{*}(M)$.
Conversely, if the geodesic flow is integrable, many quantum
limits may exist and the eigenfunctions are said to exhibit
``scarring''.  E.g., if $M=\R^2/2 \pi\Z^2$ is a flat torus and $a \in \Z$,
then $\psi_{a}(x,y) = \cos( ax) \cos(  y)$ is an
eigenfunction with eigenvalue $a^{2}+1$, and clearly
$|\psi_{a}(x,y)|^{2} \overset{*} \to \cos^{2}( y)/2$ as $a \to
\infty$.  (For a partial classification of the  set of quantum limits
on $\R^2/2 \pi\Z^2$, see \cite{J97}.) 

Now, if the flow is ergodic (``weak chaos''), Schnirelman's theorem
\cite{Sh74,Ze87,CdV85} asserts {\em Quantum Ergodicity}, namely that
the only quantum limit, provided we remove a zero density subset of
the eigenvalues, is the uniform one. However, non-uniform quantum
limits may exist along the zero density subsequence of removed
eigenvalues.  Some interesting questions for quantum ergodic systems
are thus: are there scars?  If so, how large can the exceptional set
of eigenvalues be?  Can eigenfunctions scar in position space, i.e.,
is it possible that $|\psi_{\lambda}(x)|^{2}$, along some subsequence,
weakly tends to something other than $1/\operatorname{vol}(M)$?  We
shall address these questions for the set of ``new'' eigenfunctions of
the Laplacian on a torus perturbed by a delta potential.  The
perturbation has a very small effect on the classical dynamics ---
only a zero measure subset of the set of trajectories is changed
(hence there is no classical ergodicity), yet, as was recently shown
\cite{RuU12, KuU14, Y14}, quantum ergodicity holds for the set of new
eigenfunctions.
%
(We note that this is quite different from point scatterers on tori of
the form $\R^2/\Gamma$, for $\Gamma$ a generic rectangular lattice.
Here it was recently shown \cite{KuU-scars} that quantum ergodicity
does not hold; in fact almost all new eigenfunction exhibit strong
momentum scarring, cf. Section~\ref{sec:discussion}.)

\subsection{Toral point scatterers}
The point scatterer, or the Laplacian perturbed with a delta potential
(also known as a ``Fermi pseudopotential''), is a popular ``toy
model'' for studying the transition 
between chaos and integrability in quantum chaos. 
With $\mathbb{T}^{d} :=\mathbb{R}^{d}/2\pi\mathbb{Z}^{d}$ for $d=2$ or
$d=3$, let $\alpha \in \R$ denote the ``strength'' of a delta
potential placed at some point $x_{0}\in \mathbb{T}^d$; the formal
operator
\[
%
-\Delta+\alpha \cdot \delta_{x_{0}} 
\]
can then be realized using von Neumann's theory of self adjoint
extensions.  For $d=2,3$ there is a one parameter family of self
adjoint extensions $H_{\varphi}$, parametrized by an angle $\varphi
\in (-\pi,\pi]$, and the quantum dynamics we consider is generated by
$H_{\varphi}$.  For $d=3$ we will keep $\varphi$ fixed, but in order
to obtain a strong spectral perturbation for $d=2$ we will allow
$\varphi$ to slowly vary with the eigenvalue; in the physics literature
this is known as the ``strong coupling limit'',
cf. Section~\ref{sec:background} for more details.


The spectrum of
$H_{\varphi}$ consists of two types of eigenvalues: ``old'' and
``new'' eigenvalues. The old ones are eigenvalues of the unperturbed
Laplacian, i.e., integers that can be represented as sums of $d$
integer squares, and the old eigenfunctions are the corresponding
eigenfunctions of the unperturbed Laplacian that vanish at
$x_{0}$. The set of new eigenvalues, denoted by $\Lambda$, are all of
multiplicity $1$, and interlace between the old eigenvalues.  In fact,
the new eigenvalues are solutions of the spectral equation
\begin{equation}
\sum_{n\in\mathcal N_{d}}r_{d}(n)
\left(
\frac{1}{n-\lambda}-\frac{n}{n^{2}+1}
\right)=
C
\label{eq:wk_cpl}
\end{equation}
where 
\[
r_{d}(n):=\sum_{\substack{\xi\in\mathbb{Z}^{d}, \ |\xi|^{2}=n}}1
\]
is the number of ways to represent $n$ as a sum of $d$ squares, 
$$
\mathcal{N}_{d} := \{ n \in \Z : r_{d}(n)>0 \},
$$
and 
$$
C = C(\varphi) := \tan(\varphi/2) \cdot \sum_{n} r_{d}(n)/(n^{2}+1).
$$
is allowed vary with $\lambda$ when $d=2$.

For $\lambda \in \Lambda$ a new eigenvalue, the corresponding
eigenfunction is then given by  the Green's functions $G_{\lambda}= \left(
  \Delta+\lambda \right)^{-1}\delta_{x_{0}}$,
with   $L^{2}$-expansion
$$
  G_{\lambda}(x)=
-\frac{1}{4\pi^{2}}\sum_{\xi\in\mathbb{Z}^{d}}\frac{\mbox{exp}
  (-i\xi\cdot x_{0})}{|\xi|^{2}-\lambda}e^{i\xi\cdot  x}.
$$
We remark that the delta potential introduces singularities at
$x_{0}$; as $x \to x_{0}$, we have the asymptotic (for some $a \in
\R$)
$$
G_{\lambda}(x) = 
\begin{cases}
a
\left( \cos(\varphi/2) \cdot \frac{\log |x-x_{0}|}{2 \pi} + \sin
(\varphi/2) \right) + o(1)
& \text{for $d=2$,}
\\
a
\left( \cos(\varphi/2) \cdot \frac{-1}{ 4 \pi |x-x_{0}|} +
  \sin(\varphi/2) \right) + o(1) 
& \text{for $d=3$.}
\end{cases}
$$
Note that $\varphi=\pi$ gives the unperturbed Laplacian; in
what follows we will assume that $\varphi \in (-\pi,\pi)$.

We can now formulate our first result, namely that some eigenfunctions 
strongly localize in the momentum representation in dimension
three. For $l\in\mathcal{N}_{3}$ let \[
\Omega(l):=\left\{
{\xi}/{|\xi|}\in\mathbb{S}^{2}:\xi\in\mathbb{Z}^{3},|\xi|^{2}=l
\right\}
\]
be the projection of the lattice points of distance $\sqrt l$ from the
origin onto the unit sphere, and 
let $\delta_{\Omega(l)}$ denote the distribution defined by 
\[
\delta_{\Omega(l)}(f):=\frac{1}{r_{3}(l)}
\sum_{\substack{\xi\in\mathbb{Z}^{3}\\|\xi|^{2}=l}}f
\left(
\frac{\xi}{|\xi|}
\right)
, \quad \quad \text{for $f \in C^{\infty}(\mathbb{S}^{2})$}
\]
(we can view it as the uniform probability measure on the points of
$\Omega(l)$), and let $\nu$ denote the uniform measure on
$\mathbb{S}^{2}$.
\begin{thm}
\label{thm:3dscar}
Let $\mathbb{T}^{3}=\mathbb{R}^{3}/2\pi\mathbb{Z}^{3}$,
$x_{0}\in\mathbb{T}^{3}$ and let $\Lambda$ be the set of ``new''
eigenvalues of the point scatterer, that is
$\Lambda=\rm{Spec}(H_{\varphi})\setminus\mathcal{N}_{d}$.
For $\lambda\in\Lambda$, let $g_{\lambda}\in
L^{2}(\mathbb{T}^{2})$ denote the $L^{2}$-normalized eigenfunction
with eigenvalue $\lambda$. Then for any $l\in\mathcal{N}_{3}$ there
exists an infinite subset
$\Lambda_{l}\subset\Lambda$, and $a\in[\frac{1}{2},1]$ such that for any 
pure momentum observable $f\in C^{\infty}(\mathbb{S}^{2})$
\begin{equation}
  \label{eq:scar_mmnt_dim3}
  \lim_{\lambda\in\Lambda_{l}}\langle
\Op(f)g_{\lambda},g_{\lambda}\rangle=a \cdot \delta_{\Omega(l)}(f)+
(1-a) \cdot \nu(f)
\end{equation}
That is, the pushforward of the quantum limit along this
sequence to the momentum space is a convex sum of the normalized sum
of delta measures on the finite set $\Omega(l)$, and the uniform
measure, with at least half the mass on the singular part --- there is
{\bf strong} scarring in the momentum representation.
\end{thm}

In dimension $2$, when $\varphi$ is fixed, \eqref{eq:wk_cpl} is often
referred to as the ``weak coupling limit'', and almost all new
eigenvalues remain close to the old eigenvalues (cf. \cite{RuU14}).
To find a model which exhibits level repulsion, Shigehara \cite{Shi94}
and later Bogomolny and Gerland \cite{BGS01} considered another
quantization, sometimes referred to as the ``strong coupling
limit''. One way to arrive at this quantization is by considering
energy levels in a window around a given eigenvalue: e.g., for
$\eta \in (131/146,1 )$ the new eigenvalues are defined to be solutions of
\begin{equation}
\sum_{\substack{n\in\mathcal N_{d}\\|n-n_{+}(\lambda)|<n_{+}(\lambda)^{\eta}}}r_{d}(n)
\left(
\frac{1}{n-\lambda}-\frac{n}{n^{2}+1}
\right)=
0,
\label{eq:strng_cpl}
\end{equation}
where  $n_{+}(\lambda)$ is  the smallest
element of $\mathcal{N}_{2}$ that is larger than $\lambda$.
It is convenient to consider both couplings simultaneously; we may do this
by letting
\begin{equation}
\label{def:cpl_fn}
F(\lambda)=
\begin{cases}
\text{Constant} & \mbox{(weak coupling)}\\
\sum\limits_{\substack{n\in\mathcal{N}_{2}\\|n-n_{+}(\lambda)|\geq
        n_{+}(\lambda)^{\eta}}}\hspace{-25pt}r_{2}(n)\left(
\frac{1}{n-\lambda}-\frac{n}{n^{2}+1}
\right) & \mbox{(strong coupling)}
\end{cases}
\end{equation}
and then rewriting the spectral equation as
\begin{equation}
\sum_{n\in\mathcal N_{2}}r_{2}(n)
\left(
\frac{1}{n-\lambda}-\frac{n}{n^{2}+1}
\right)=F(\lambda)
\label{eq:bth_cpl}
\end{equation} 

Our next result, valid for both the weak and strong coupling limit in
dimension two, is the existence of a zero density subsequence
exhibiting non-uniform quantum limits in the momentum, as well as the
position representation.
\begin{thm}\label{thm:2dscar}
  Let $\mathbb{T}^{2}=\mathbb{R}^{2}/2\pi\mathbb{Z}^{2}$,
  $x_{0}\in\mathbb{T}^{2}$ and let $\Lambda$ be the set of new
  eigenvalues of the point scatterer, that is
  $\Lambda=\rm{Spec}(H_{\varphi})\setminus\mathcal{N}_{2}$.
  For $\lambda\in\Lambda$, let $g_{\lambda}\in
  L^{2}(S^{*}\mathbb{T}^{2})$ be the $L^{2}$-normalized eigenfunction
  with eigenvalue $\lambda$.
  \begin{enumerate}
  \item \label{thm:2dscar_momentum}
 There exists an infinite subset
  $\Lambda_{m}'\subset\Lambda$ such that the pushforward of the quantum
  limit along this sequence to momentum space has  positive
  mass on a
  finite number of 
  atoms (``strong momentum scarring'').
\item \label{thm:2dscar_position}
There exists an infinite subset $\Lambda_{p}'\subset\Lambda$ such that the
  pushforward of the quantum limit along this sequence to position
  space has a nontrivial non-zero Fourier coefficient (``position scarring'').
  \end{enumerate}


Furthermore, we may take $\Lambda_{p}'=\Lambda_{m}'$.
\end{thm}

We remark that for $d=2$ almost half the mass is carried on the
singular part in the momentum representation --- for any $\epsilon>0$
the singular part has mass at least $1/2-\epsilon$
(cf. Remark~\ref{rem:half-mass-singular}).
%

In order to quantify how common scars are we need some further
notation.  For $d=2,3$, let $N_{d}(x)$ denote the counting function
(``Weyl's law'') for the number of new eigenvalues $\lambda \leq x$.
For $d=3$, $N_{3}(x) \sim x$ and, as the eigenvalues that give rise to
scars are essentially powers $4^{l}$, the exceptional subset is of
size $x^{o(1)}$ and thus very sparse.  For $d=2$, $N_{2}(x) \sim
x/\sqrt{\log x} = x^{1 -o(1)}$, and our construction of eigenfunctions
that scar both in position and momentum is a subset with counting
function of size $ x^{1/2-o(1)}$ --- hence  fairly rare.
However, if we restrict ourselves to scarring only in the momentum
representation, we can use some recent results by Maynard
\cite{Maynard-dense-clusters} to show that scarred eigenvalues are in
fact quite common.
\begin{thm}
\label{thm:large-density-scars}
In dimension two there exists a subset $\Lambda''\subset\Lambda$ such
that the pushforward of the quantum limit along $\Lambda''$ scars in
momentum space, and 
$$
|\{ \lambda \in \Lambda'' : \lambda \leq x \}|
\gg
x/(\log x)^{A}
$$
for some $A>1$.
\end{thm}

\subsection{Discussion}
\label{sec:discussion}
%
In \cite{Se90} \v{S}eba proposed quantum billiards on rectangles with
irrational aspect ratio, perturbed with a delta potential, as a
solvable singular model exhibiting wave chaos; in particular that the
level spacings should be given by random matrix theory (GOE).
%
\v{S}eba and {\.Z}yczkowski later noted \cite{SeZy91} that the level
spacings were not consistent with GOE, in particular large gaps are
much more frequent (essentially having a Poisson distribution tail.)
Shigehara subsequently found \cite{Shi94}  that level repulsion
is only present in the strong coupling limit.  Recently Rudnick and
Uebersch\"ar proved \cite{RuU14}, in dimension two, that the level
spacing for the weak coupling limit is the same as the level spacings
of the unperturbed Laplacian (after removing multiplicities).  This in
turn is conjectured to be Poissonian, and we note that a natural
analogue of the prime $k$-tuple conjecture for integers that are sums
of two squares can be shown to imply Poisson gaps \cite{SOTS}.  In
\cite{RuU14} the three dimensional case was also investigated and the
mean displacement between new and old eigenvalues was shown to equal
half the mean spacing.


In \cite{RuU12}, Rudnick and Uebersch\"ar proved a position space
analogue of Quantum Ergodicity for the new eigenfunctions: there
exists a full density subset of the new eigenvalues such that as
$\lambda \to \infty$ along this subset, the only weak limit of
$|\psi_\lambda(x)|^{2}$ is the uniform measure on $\mathbb{T}^2$.
Further, in \cite{KuU14} the first author and Ueberschär proved an
analogue of Quantum Ergodicity: there exists a full density subset of
the new eigenvalues such that the only quantum limit along this subset
is the uniform measure on the full phase space (i.e., the unit
cotangent bundle $S^{*}(\mathbb{T}^{d})$.)  This result was later
shown to hold also for $d=3$ by Yesha \cite{Y14}; already in
\cite{Y13} he showed that {\em all} eigenfunctions equidistribute in
the position representation.

For irrational tori, Keating, Marklof and Winn proved in \cite{KeMW10}
that there exist non-uniform quantum limits (in fact, strong
momentum scarring was already
observed  in \cite{BKW03}), assuming a spectral 
clustering condition implied by the old eigenvalues having Poisson
spacings (which in turn follows from the Berry-Tabor conjecture.)
Recently the first author and Ueberschär unconditionally showed
\cite{KuU-scars} that for tori having diophantine aspect ratio,
essentially {\em all} new eigenfunctions strongly scar in the momentum
representation. 
Recently Griffin showed \cite{Gr15} that similar results hold for
Bloch eigenmodes (i.e., non-zero quasimomentum) for periodic point
scatterers in three dimensions, provided a certain Diophantine
condition on the aspect ratio holds.




\subsection{Scarring and QUE for some other models}
For Quantum Ergodic systems almost all eigenfunctions equidistribute,
but in general not much is known about the (potential) subset of
expectional eigenfunctions giving non-uniform quantum limits.
%
In some cases Quantum Unique Ergodicity is known to hold; notable examples are
Hecke eigenfunctions on modular surfaces \cite{Lin06,S10} and
``quantized cat maps'' \cite{KuRu00, Kel10}. For these models there
exist large commuting families of ``Hecke symmetries'' that also
commute with the quantized Hamiltonian, and it is then natural to
consider joint eigenfunctions of the full family of commuting
operators.
%
Other examples arise when the underlying classical dynamics is
uniquely ergodic, QUE is then ``automatic'', e.g., see
\cite{Ro06,MR00}.

On the other hand there are Quantum Ergodic systems exhibiting
scarring. 
For example, if Hecke symmetries are not taken into account, quantized cat
maps can have very large spectral degeneracies. Using this, Faure,
Nonnenmacher and de-Bievre (\cite{FNdB03}) proved that scars occur in
this model.  For higher dimensional analogues of cat maps, Kelmer
found a scar construction not involving spectral degeneracies, but
rather certain invariant rational isotropic subspaces
\cite{Kel10,Kel11}.

We also note that Berkolaiko, Keating, and Winn has shown
\cite{BKW03,BKW04} that simultaneous momentum and position scarring
can occur for quantum star graphs, e.g., for certain star graphs with
a fixed (but arbitrarily large) number of bonds, there exists quantum
limits supported only on two bonds.

Another way to construct scars is to use ``bouncing ball quasimodes''.
For example, functions of the form $\psi_n(x,y) = f(x) \sin( ny)$ are
approximate Laplace eigenfunctions on a stadium shaped domain (say
with Dirichlet boundary conditions), and semiclassically localize on
vertical periodic trajectories.  Hassell showed \cite{H10} that for a
generic aspect ratio stadium, there are few eigenvalues near $n$ and
hence 
$\psi_n$ overlaps strongly with an eigenfunction $\phi_n$ with
eigenvalue near $n$, which then also must partially localize on
vertical periodic trajectories.  The number of ``bouncing ball 
eigenfunctions'' having eigenvalue at most $E$ grows (at most) as
$E^{1/2+o(1)}$, to be compared with the Weyl asymptotic $c \cdot E$;
hence these scarred eigenstates are fairly rare.
In \cite{BSS97,T97,LBK12} the asymptotic behaviour of sets of bouncing
ball eigenfunctions for some ergodic billiards were considered.
Interestingly, for the stadium billiard it was argued that the number
of scarred bouncing ball eigenfunctions, with eigenvalue at most $E$,
are much more numerous, namely of order $ E^{3/4}$ (again to be
compared the Weyl asymptotic $c \cdot E$.)  In fact, in \cite{BSS97}
it was argued that given any $\delta \in (1/2,1)$, there exists a Sinai
type billiard whose bouncing ball eigenfunction count 
is of order $c_{\delta} \cdot E^{\delta}$.



\subsection{Outline of the proofs}
\label{sec:outline-proofs}
The proofs are based on finding new eigenvalues $\lambda$ that are
quite near certain old eigenvalues. After rewriting equation
\eqref{eq:bth_cpl} as
\begin{equation}
  \label{eq:bth_cpl_2}
\frac{r_{d}(m)}{m-\lambda}-\frac{m}{m^{2}+1}+H_{m}(\lambda)=0,
\end{equation}
where
\[
H_{m}(\lambda) :=
\sum_{\substack{n\neq
    m\\n\in\mathcal{N}_{d}}}r_{d}(n)
\left(
\frac{1}{n-\lambda}-\frac{n}{n^{2}+1}
\right)-F(\lambda),
\]
we show that for any $m\in\mathcal{N}_{d}$ there exists a new
eigenvalue $\lambda$ such that $|m-\lambda|
\ll\sqrt{r_d(m)/H_{m}'(m)}$ (though it should be emphasized that we do not know
whether $\lambda >m$ or $\lambda<m$).  
We then find a sequence of integers $m$ such that both $r_{d}(m)$ and
$\sqrt{r_{d}(m)/H_{m}'(m)}$ are bounded, and thus get a control on the
distance of a new eigenvalue from these $m$.  (To find such $m$ we use the lower bound
sieve methods when $d=2$; for $d=3$ we find integers $m$ for which
the representation number $r_{3}(m)$ is very small.)  We conclude by using an
explicit description for the relevant eigenfunctions to compute the
limits in the theorems.

The paper is organized as follows: In Section \ref{sec:background} we
set the necessary background for the point scatterer model, then give some
number theoretic background, and in
Section~\ref{sec:preliminaries} we prove some
auxiliary analytic and
number theoretic results needed in the proofs of our main theorems. In
Sections \ref{sec:proof-theor-refthm:3} and \ref{sec:proof-2dim} we
prove Theorems \ref{thm:3dscar} and \ref{thm:2dscar}, and
Section~\ref{sec:version-2} contains the proof of
Theorem~\ref{thm:large-density-scars}. 
\subsection*{Acknowledgements} We would like to thank Ze\'ev Rudnick and
Henrik Uebersch\"ar for helpful discussions about this work.

 \section{Background}
 \label{sec:background}
 In this section we briefly review some results and definitions about
 point scatterers and give a short number theoretic background.
\subsection{Point scatterers on the flat torus}
We begin with the point scatterers, and recall the definition and
properties of the quantization of observables (see  \cite{RuU12,KuU14} for more
 details; further background can be found in \cite{U14,Y14}.)
 \subsubsection{Basic definitions and properties}
 \label{sec:basic-defin-prop}
For $d=2,3$ we consider the restriction of the Laplacian $-\Delta$ on 
\[
D_{0}:=C^{\infty}_{0}(\mathbb{T}^{d}\setminus\{x_{0}\})
\]
The restriction is symmetric though not self-adjoint, but by von
Neumann's theory of self adjoint extensions there exists a
one-parameter family of self-adjoint extensions; for
$\varphi\in(-\pi,\pi]$ there exists a self-adjoint extension
$H_{\varphi}$, where the case $\varphi=\pi$ corresponds to the
unperturbed Laplacian. The spectrum of $H_{\varphi}$ consists of two
types of eigenvalues and eigenfunctions:
\begin{enumerate}
\item Eigenvalues of the unperturbed Laplacian, and the corresponding
  eigenfunctions that vanish at $x_{0}$. The multiplicities of the new
  eigenvalues are reduced by 1, due to the constraint of vanishing at
  $x_{0}$.
\item New eigenvalues $\lambda\in\mathbb{R}$ satisfying the equation
  \begin{equation}
    \label{eq:evalue}
  \sum_{n\in\mathcal{N}_{d}}  r_{d}(n)
  \left(
\frac{1}{n-\lambda}-\frac{n}{n^{2}+1}
\right)
=c_{0}\tan
\left(
\frac{\varphi}{2}
\right)
  \end{equation}
\end{enumerate}
For $\lambda\in\mathbb{R}$ satisfying \eqref{eq:evalue}, the corresponding
Green's function
\begin{multline}
  \label{def:greens function}
  G_{\lambda}(x,x_{0})=
  \left(
\Delta+\lambda
\right)^{-1}\delta_{x_{0}}=\\
-\frac{1}{4\pi^{2}}\sum_{\xi\in\mathbb{Z}^{d}}\frac{\mbox{exp}
  (-i\xi\cdot x_{0})}{|\xi|^{2}-\lambda}e^{i\xi\cdot
  x},\quad x\neq x_{0}     
\end{multline}
is an eigenfunction, and
\begin{equation}
  \label{eq:def-g-lambda}
 g_{\lambda}(x,x_{0}):=\frac{G_{\lambda}(x,x_{0})}{\|G_{\lambda}\|}=\frac{\displaystyle\
  \sum_{\xi\in\mathbb{Z}^{d}}\frac{\mbox{exp}(-i\xi\cdot
     x_{0})}{|\xi|^{2}-\lambda}e^{i\xi\cdot x}}{\displaystyle 
\left(
   \sum_{n\in\mathcal{N}_{d}}\frac{r_{d}(n)}{|n-\lambda|^{2}}
\right)^{1/2}
}   
\end{equation}
is an $L^{2}$-normalized eigenfunction.

\subsubsection{Strong coupling}
\label{sec:strong-coupling}
In \cite{RuU12} Rudnick and Ueberschär showed that for $d=2$ the set
of ``new'' eigenvalues ``clump'' with the Laplace eigenvalues, and in
fact the eigenvalue spacing distribution coincides with that of the
Laplacian. In \cite{Shi94} Shigehara, and in \cite{BGS01} Bogomolny,
Gerland and Schmit considered another type of quantization, with the
intent of finding a model that exhibits level repulsion. This
quantization is sometimes referred to as the ``strong coupling''
(compared to the ``weak coupling'' given by equation
\eqref{eq:evalue}). One way of arriving at this quantization is by
truncating the summation in \eqref{eq:evalue} outside an energy window
of size $O(\lambda^\eta)$ for any fixed $\eta>131/146$. This leads to
the following spectral equation for the new eigenvalues:
\begin{equation}
\label{eq:evalue_strong}
\sum_{\substack{n\in\mathcal{N}_{2}\\|n-n_+(\lambda)|< n_+(\lambda)^\eta}}  r_{2}(n)
  \left(
\frac{1}{n -\lambda}-\frac{n}{n^{2}+1}
\right)
=c_{0}\tan
\left(
\frac{\varphi}{2}
\right).
\end{equation}
\subsubsection{Quantization of observables} Given a smooth observable
$a(x,\xi)$ on $S^{*}(\mathbb{T}^{d})\simeq \mathbb{T}^{d}\times
\mathbb{S}^{d-1}$ we define the quantization of it as a
pseudo-differential operator $\Op(a):C^{\infty}(\mathbb{T}^{d})\to
C^{\infty}(\mathbb{T}^{d})$. We refer the reader to
\cite{KuU14} for details on the 2 dimensional case, and \cite{Y14} for
the 3 dimensional case. We are mainly interested in either pure
momentum, or pure position observables, that is $a(x,\xi)=a(\xi)\in
C^{\infty}(\mathbb{S}^{d-1})$, or $a(x,\xi)=a(x)\in
C^{\infty}(\mathbb{T}^{d})$ respectively; this considerably simplifies
the discussion of quantizing observables.  Namely, given $f(x)\in
C^{\infty}(\mathbb{T}^{d})$, the action of a pure position
observable  $a=a(x)\in C^{\infty}(\mathbb{T}^{d})$ is given by 
\begin{equation}
  \label{eq:def-quantization-pos}
(\Op(a)f)(x)=a(x)f(x),
\end{equation}
whereas the action of a pure momentum observable
$a=a(\xi)\in C^{\infty}(\mathbb{S}^{d-1})$ is given by
\begin{equation}
  \label{eq:def-quantization-mom}
(\Op(a)f)(x)=\sum_{v\in\mathbb{Z}^{d}}a
\left(
\frac{v}{|v|}
\right)\widehat{f}(v)e^{iv\cdot x};
\end{equation}
in particular, for pure momentum observables we have
\begin{equation}
  \label{eq:def-matrix-elem}
\langle\Op(a)f,f\rangle=\sum_{v\in\mathbb{Z}^{d}}a
\left(
\frac{v}{|v|}
\right)|\widehat{f}(v)|^{2}.
\end{equation}

\subsection{Number theoretic background}
\label{sec:numb-theor-backgr}
\subsubsection{Integers that are sums of $2$ or $3$ squares}
\label{sec:sums-d-squares}

We begin with a short summary about integers that can be represented as
sums of $d$ squares for $d=2$ or 3. 
\begin{description}[leftmargin=0pt,itemindent=0pt]
\item[Sums of 2 squares]
It is well known (e.g., see \cite{lattice-points-on-circles}) that
$r_{2}(n)$ is determined by the prime factorization of $n$.  If we
write 
$$n=2^{a_{0}}p_{1}^{a_{1}}\dots p_{r}^{a_{r}}q_{1}^{b_{1}}\dots
q_{l}^{b_{l}},
$$
 where the $p_{i}$'s are primes all $ \equiv 1\pmod4$,
and the $q_{i}$'s are primes all $\equiv 3\pmod4$, then $n$ is a sum of two
squares if and only if all the $b_{i}$ are even, and
$r_{2}(n)=4d(p_{1}^{a_{1}}\dots p_{r}^{a_{r}})$, where $d(\cdot)$ is
the divisor function.
\\

\item[Sums of 3 squares]
\label{sec:sums-3-squares}
 For $d=3$, any number $n$ that is not
of the form $n=4^{a}n_{1}$, where $4 \nmid n_{1}$ and
$n_{1}\not\equiv7\pmod8$ can be represented as a sum of 3
squares. Moreover, $r_{3} (4n)=r_{3}(n)$ for any $n\in\mathbb{Z}$, and
if we let $R_{3}(n)$  denote the number of primitive representation of
$n$ as a sum of 3 squares (that is the number of ways to write
$n=x^{2}+y^{2}+z^{2}$ with $x,y,z$ coprime), we can relate $r_{3}(n)$
to class numbers of quadratic imaginary fields as follows
(cf. \cite[Theorem 4, p. 54]{Gr_book}):
\begin{equation}
  \label{eq:3repsformula}
r_{3}(n)=\sum_{d^{2}|n}R_{3}(\frac{n}{d^{2}}), \quad
R_{3}(n)=\pi^{-1}G_{n}\sqrt{n}L(1,\chi_n),
\end{equation}
where 
$$
G_{n}=
\begin{cases}
  0 & n\equiv 0,4,7\pmod8\\
 16 & n\equiv 3\pmod8\\
 24 & n\equiv 1,2,5,6\pmod8
\end{cases}.
$$
and $\chi_n(m)=(-4n/m)$ is the Kronecker symbol. 
By a celebrated theorem of Siegel, for any $\epsilon > 0$, 
$L(1,\chi_n) \gg_{\epsilon} n^{-\epsilon}$
and thus, for $n \not \equiv 0,4,7 \mod 8$,
\begin{equation}
  \label{eq:siegel-bound}
r_{3}(n) \geq R_{3}(n) \gg_{\epsilon} n^{1/2-\epsilon}.
\end{equation}
%
Further, given an integer $n$ that is a 
sum of 3 squares, let 
\[
\Omega(n):=
\left\{ 
\frac{(x,y,z)}{\sqrt{n}}: (x,y,z) \in \Z^3,    x^{2}+y^{2}+z^{2} = n
\right\}\subset\mathbb{S}^{2}.
\]
Fomenko-Golubeva and Duke showed (see
\cite{golubeva-fomenko-lattice-points-on-spheres87, Du88}, or
\cite[Lemma 2]{DuSchu90}) that the sets $\Omega(n)$ equidistribute in
$\mathbb{S}^{2}$ as $r_3(n)\to\infty$ (or equivalently $n_1\to\infty$)
inside $\mathcal{N}_{3}$. Namely, there exists $\alpha>0$, such that
for any spherical harmonic $Y(x)$, there is significant cancellation
in the Weyl sum
\[
W_{Y}(n):=\sum_{\xi\in\Omega(n)}Y(\xi),
\]
in the sense that
\begin{equation}
  \label{eq:Weyl-sum-cancelation}
W_{Y}(n)\ll n_{1}^{1/2-\alpha} \ll n^{1/2-\alpha}
\end{equation}
where the implied constant is independent of $n$.
\end{description}
\subsection{Sieve method results}
\label{sec:sieve-method-results}
We list below some  sieve results that show the existence
of various infinite sequences of integers with bounded number of prime
divisors. We first recall a few definitions. A positive integer $n$ is
called {\it{$r$-almost prime}} if $n$ has at most $r$ prime
divisors. We denote by $P_{r}$ the set of all $r$-almost prime
integers. A finite set of polynomials $\mathcal F=
\left\{
F_{1}(x),\dots,F_{k}(x)
\right\}\subset\mathbb{Z}[x]
$ is called {\it{admissible}} if $F(x):=\prod_{i=1}^{k}F_{i}(x)$ has no
fixed prime divisors, that is the equation $F(x)\equiv0\pmod p$ has
less that $p$ solutions for any prime $p$.
 The followoing theorem combines results from \cite{halberstam-richert,FI10,Maynard-dense-clusters}:
\begin{thm}
\label{thm:sieve}
Let $F_{1}(x),\dots,F_{k}(x)\in\mathbb{Z}[x]$ ($k\geq1$) be a finite
admissible set of irreducible
polynomials, and let $F(x):=F_{1}(x)\cdots F_{k}(x)$. Let $G$ denote
the degree of $F$. Then,
\begin{enumerate}[leftmargin=25pt,itemindent=0pt]
\item 
  \label{thm:halberstam-richert}\cite[Theorem 10.4]{halberstam-richert}
 There exists an integer $R(k,G)$,
such that for any $r>R(k,G)$, as $x\to\infty$,
\[
\#
\left\{
n\in\mathbb{Z},n\leq x:F(n)\in P_{r}
\right\}\gg\frac{x}{\log^{k}x}
\]
\item 
 \label{thm:iwaniec-friedlander}\cite[Theorem 25.4]{FI10} If $k=1$,
 then one can take $R(k,G)=G+1$, and therefore as $x\to\infty$
\[
\#
\left\{
n\leq x:F(n)\in P_{G+1}
\right\}\gg\frac{x}{\log x}
\]
\item
\label{prop:maynard-implication}
\cite[Theorem~3.4]{Maynard-dense-clusters}
If $k$ is large enough, and $F_{1}(x),\dots,F_{k}(x)$ are all {\bf{\em{linear}}}
with positive coefficients, then as $x \to \infty$
$$
|\{ n \leq x : \text{ at least two of the $F_{i}(n)$, $1\leq i \leq
  k$, are prime} \}| \gg \frac{x}{(\log x)^{k}}
$$
\end{enumerate}
\end{thm}
\section{Auxiliary results}
\label{sec:preliminaries}
Before proceeding to the proofs of the main theorems, we begin with a
few auxiliary results. For the benefit of the reader we note that
\S{\ref{sec:nearby-zeros}} is relevant for all theorems, Lemma
\ref{lem:almost-prime-elts-in-S} is relevant for the proof of Theorem
\ref{thm:2dscar}, and Lemma
\ref{lem:one-good-tuple}, and Proposition \ref{prop:r2-boost} are relevant for the
proof of Theorem \ref{thm:large-density-scars}.

\subsection{Nearby zeros}
\label{sec:nearby-zeros}
The following simple result will be crucial in finding integers $m$ in
the old spectrum for which there exist a nearby new eigenvalue $\lambda$.
\begin{lem}
\label{lem:nearby-zeros}
  Let $I$ be a closed symmetric interval containing zero, and let $f$
  be $C^{1}$ function on $I$.  Let $A>0$  be a real
  number, and assume that $B := \min_{\delta \in
    I} f'(\delta) >0$.  If $\sqrt{A/B} \in I$ there exists $\delta_{0} \in [-
  \sqrt{A/B},  \sqrt{A/B}]$ such that
$$
f(\delta_{0}) = A/\delta_{0}.
$$
\end{lem}
\begin{proof}
Let $I^{+} = I \cap [0,\infty]$, and let $I^{-} = I \cap [-\infty, 0]$.
For $\delta \in I^{+}$, we have $f(\delta) \geq f(0) + B\delta$.
Similarly, for $\delta 
\in I^{-}$, $f(\delta) \leq f(0) + B\delta$.  Thus, since $A,B>0$, if
$$
f(0) + B \delta_{1} = A/\delta_{1}
$$
for $\delta_{1} \in I^{+}$, there exists $\delta_{0} \in
[0,\delta_{1}]$ such that $f(\delta_{0}) = A/\delta_{0}$.  Similarly, if
$f(0) + B\delta_{1} = A/\delta_{1} $ for $\delta_{1} \in I^{-}$,
there exists $\delta_{0} \in [\delta_{1},0]$ such that $f(\delta_{0})
= A/\delta_{0} $.

To conclude the proof it is enough to show that 
$$
f(0) + B\delta = A/\delta 
$$
has a solution in $[-\sqrt{A/B}, \sqrt{A/B}]$, but this is clear since
$
B \delta^{2} + f(0)\delta -A = 0
$
has at least one root $\delta_{1}$ for which $|\delta_{1}| \leq \sqrt{A/B}$.

\end{proof}

\subsection{Sequences of sums of two squares}
\label{sec:sequences-sums-two}

\begin{lem}
\label{lem:almost-prime-elts-in-S}
Given $\gamma \in (0,1/10)$ there exists an infinite set $\mathcal M_{\gamma} \subset \mathcal{N}_{2}$
with the following properties: $\forall m\in\mathcal{M}_{\gamma}$
$$
m=n^{2}+1
$$
for some $n \in \Z^+$, 
\begin{equation}
  \label{eq:r2mbound}
r_{2}( m)  \leq 32,
\end{equation}
and
\begin{equation}
  \label{eq:r2mplus4bound}
r_{2}( m+3)  \geq 10 r_{2}(m)/\gamma^{2}.
\end{equation}
\end{lem}
\begin{proof}

We apply part \eqref{thm:iwaniec-friedlander} of Theorem \ref{thm:sieve} in the following
  setting: For $K\in\mathbb{Z}$ define 
\[
P(K)=\displaystyle\prod_{\substack{p\leq
      K\\p\equiv1\pmod4}}p,
\] 
and $r(K) \in \Z$ solving the following congruences:
  \begin{eqnarray*}
    \label{eq:1}
&& r(K)^{2}+4\equiv0\pmod{p}{\mbox{ if }} p\equiv1\pmod4,p\leq K\\
&&r(K) \equiv 0 \pmod2 
  \end{eqnarray*}
Note that the latter equation has a solution by the Chinese remainder
theorem together with $-4$ being a
quadratic residue  for any prime $p\equiv1\pmod4$. We may take $0 \leq
r(K) < 2P(K)$ but any fixed choice will suffice.
Let $f(x)=x^{2}+1$, and let
$x_{K}(n)=2P(K)\cdot n+r(K)$. The polynomial $F(n):=f(x_{K}(n))$ satisfies 
all the conditions of the theorem (it is irreducible and no prime
divides all coefficients), and therefore there are infinitely many $n$
such that $F(n)$ has at most 3 prime factors, and in particular
$r_{2}(F(n))\leq 32$. By construction, $F(n)+3 =x_{k}(n)^{2}+2^{2}
\equiv 0\pmod p$ for $p\leq K$
and  $p\equiv1\pmod4$, hence $r_{2}(F(n)+3)\geq
4\cdot2^{\pi(K;1,4)}$, where $\pi(K;1,4)$ is the number of primes occurring
in the product defining $P(K)$. By choosing $K$ appropriately, we get that  
\[
r_{2}(F(n)+3)\geq 4d(P(K))\geq 2^{\pi(K;1,4)+2}\geq
320/\gamma^{2}\geq 10 r_{2}(F(n))/\gamma^{2}
\]
\end{proof}
\begin{lem}
\label{lem:one-good-tuple}
Given $H,R \geq 2$ there exists elements $0< a_{1} < a_{2}
\ldots < a_{H}$ in $\mathcal{N}_2$ such that $a_{H}-a_{1} < H^{2}$,
$0< r_{2}(a_{1}) < \ldots < r_{2}(a_{H}) \ll_{H} R^{H}$, and
$$
r_{2}(a_{i+1} ) > R \cdot r_{2}(a_{i})
$$
holds for some $0< i < H$.
\end{lem}
\begin{proof}
Define 
$$
Q_{1} = Q_{1}(H) := \prod_{p < 2H} p^{E_{p}}
$$
where the exponents $E_{p}$ are chosen as follows: let $E_{p}=1$ if $p
\equiv 3 \mod 4$, otherwise let $E_{p}$ be the minimal integer so that
$p^{E_{p}} > H^{2}$.

Further, let $q_{1}< q_{2} < \ldots < q_{H}$ be primes congruent to
$1 \mod 4$, chosen so that $q_{1} > 2H$, and 
given integer exponents $e_{1},\ldots, e_{H} \geq 1$, define
$$
Q_{2} = Q_{2}(e_{1},e_{2},\ldots,e_{H}):= \prod_{i \leq H} q_{i}^{e_{i}}
$$
and finally let $Q := Q_{1} \cdot Q_{2}$.

By the Chinese remainder theorem we may find $\gamma \mod Q$ such that
the following holds: 
\begin{equation}
  \label{eq:congruence-one}
\gamma \equiv 0 \mod p^{E_{p}} \quad \text{if $p \equiv 1,2 \mod 4$ and $p <2H$},
\end{equation}
\begin{equation}
  \label{eq:congruence-two}
\gamma \equiv 1 \mod p \quad \text{if $p \equiv 3 \mod 4$ and $p < 2H$}
\end{equation}
and for each prime  $q_{i} | Q_{2}$ so that
\begin{equation}
  \label{eq:congruence-three}
q_{i}^{e_{i}}  || (\gamma^{2}+i^{2})
\end{equation}

Letting $d_{i} = (Q_{1}^{2},\gamma^{2}+i^{2})$ we  define polynomials
$G_{i} \in \Q[t]$ by
$$
G_{i}(t) := ((Q t  + \gamma)^{2} + i^{2})/(q_{i}^{e_{i}} d_{i}),
\quad i = 1,2, \ldots, H.
$$  
By definition, $d_{i} | \gamma^{2}+i^{2}$, and
\eqref{eq:congruence-three} implies 
that $q_{i}^{e_{i}} | \gamma^{2}+i^{2}$. Thus, since $(Q_{1},Q_{2})=1$
implies that $(q_{i},d_{i})=1$, we find that $q_{i}^{e_{i}} d_{i} |
\gamma^{2}+i^{2}$ and consequently $G_{i}(t) \in \Z[t]$ for all $i$.

\begin{claim}
$\{G_{i}\}_{i=1}^{H}$ is an admissible set of
polynomials (i.e., $\prod_{i=1}^{H} G_{i}(x)$ does not have any fixed
prime divisors).  
\end{claim}
To prove the  claim we argue as follows:
If $p > 2H$ and all $G_{i}$ are nonconstant modulo $p$ (i.e.,
$p \nmid Q$) there are at most $2H$ residues $n$ (modulo $p$) for which
$G_{i}(n) \equiv 0 \mod p$ for some $i$.  Hence there exist $n\in \Z$
such that $\prod_{i=1}^{H} G_{i}(n) \not \equiv 0 \mod p$.

On the other hand, if $p>2H$
and $p|Q$ then $p = q_{i}$ for some $i$, and by the definition of
$G_{i}$ (in particular, recall \eqref{eq:congruence-three}), we find 
that $G_{i}(n) \not \equiv 0 \mod q_{i}$ for all $n\in \Z$.  Moreover,
if $j \neq i$,
$$
\gamma^{2} + j^{2} \equiv
\gamma^2 + i^{2} + j^{2} - i^{2} \equiv
j^{2}-i^{2} \not \equiv 0 \mod q_{i}
$$
(as $0<|i-j|<H$,  $0<i+j < 2H$ and $q_{i} > 2H$), and thus $G_{j}(n)
\not \equiv 0 \mod q_{i}$ for all $n \in \Z$.

For $p < 2H$ we argue as follows: if $p \equiv 3 \mod 4$, \eqref{eq:congruence-two}
gives that $\gamma^{2} + i^{2} \not \equiv 0 \mod p$ for all $i \in
\Z$.  Otherwise, $i^{2} \leq H^{2} < p^{E_{p}}$ by our choice of
$E_{p}$, and since $\gamma$ was chosen so that $\gamma \equiv 0 \mod
p^{E_{p}}$ (recall \eqref{eq:congruence-one}), we find that $ \gamma^{2} + i^{2}
\equiv i^{2} \not \equiv 0 \mod p^{2E_{p}}$, as $i^{2} \leq H^{2}$
and $p^{E_{p}} > H^{2}$.
Consequently
$(\gamma^{2}+i^{2})/d_{i}$ is not divisible by $p$.
The proof of the claim is concluded.

Now, given an integer $r >0$, let $P_{r}$ denote the set of integers that
can be written as a product of at most $r$ primes, as in \S\ref{sec:sieve-method-results}.
Since the polynomials $\{ G_{i}(x) \}_{i=1}^{H}$ form an admissible set,
part \eqref{thm:halberstam-richert} of Theorem \ref{thm:sieve}
implies that there exists some
$r>0$ (only depending on $H$) such that
$$
\prod_{i=1}^{H} G_{i}(n) \in P_{r}
$$
for infinitely many $n$.  Given such an $n$, let $m_{i} = G_{i}(n)$;
then each $m_{i}$ is a sum of squares that in addition has at most $r$
prime factors.  Consequently, if we set $a_{i} = m_{i} \cdot
q_{i}^{e_{i}} \cdot d_{i}$, we find that $a_{i} \in \mathcal{N}_2$ for
all $1 \leq i \leq H$, and that
$$
e_{i}+1 \leq r_{2}(a_{i})   \leq 4 C \cdot (e_{i}+1) 
$$
where $C = C(H) \geq 1$ is independent of the exponents $e_{1},
\ldots, e_{H}$.  Choosing $e_{1}, \ldots, e_{H}$ appropriately we can
ensure that $r_{2}(a_{i+1})> R \cdot r_{2}(a_{i})$ holds for all $i$,
as well as that $r_{2}(a_{H}) \ll_{H} C^{H} R^{H} \ll_{H} R^{H}$ (a
somewhat better $C$-dependency can be obtained but we shall not need
it.)

Finally, since 
$$
a_{i} = m_{i} q_{i}^{e_{i}} d_{i} = 
G_{i}(n)   q_{i}^{e_{i}} d_{i} = 
 (Qn+\gamma)^{2} + i^{2}
$$ 
we find that
$a_{H}-a_{1} = H^{2}-1 < H^{2}$ and the proof of
Lemma~\ref{lem:one-good-tuple} is concluded.

\end{proof}

The following proposition might be of independent interest --- using
the full power of \cite{Maynard-dense-clusters}, the method of the
proof in fact gives the following: given $k \geq 2$ and $R > 1$ there
exists $A>0$ such that, as $x \to \infty$, there are $\gg x/(\log
x)^{A}$ integers $n \leq x$ such that $r_{2}(n+h_{i+1}) \geq R
r_{2}(n+h_{i})$ holds for $i=1, \ldots, k-1$ and $0< h_{1} < h_{2} <
\ldots < h_{k} \ll_{k} 1$.  For simplicity we only  state and prove
it for $k=2$.
\begin{prop}
\label{prop:r2-boost}
There exist an integer $H \geq 1$ with the following property:  for
all sufficiently large $R$
there exist an integer $h  \in (0,H^{2})$ such that
\begin{multline*}
|\{ n \in \mathcal{N}_2 : n \leq x, \quad 0 < r_{2}(n) \ll R^{H},
\quad r_{2}(n+h) \geq R \cdot r_{2}(n)  \}|
\\
\gg_{R} x/(\log x)^{H}
\end{multline*}
as $x \to \infty$.
\end{prop}

\begin{proof}

By part~\eqref{prop:maynard-implication} of Theorem \ref{thm:sieve} there exists
integers $i,j$ such that $0< i < j \leq H$ with the property that
$$
|\{ n \leq x : F_{i}(n), F_{j}(n) \text{ both prime} \}|
\gg x/\log^{H} x
$$
for $\{ F_{1},F_{2},\ldots, F_{H} \}$ {\em any} admissible set of $H$
{\bf{linear}} forms, provided $H$ is sufficiently large.  For such an $H$,
and a given (large) $R$,  Lemma~\ref{lem:one-good-tuple} shows there
exists $a_{1}, \ldots, a_{H}>0$ such that
$$
r_{2}(a_{i+1}) \geq R \cdot r_{2}(a_{i}) > 0
$$
for $1 \leq i < H$, and $r_{2}(a_{H}) \ll R^{H}$.  If we define
$$
F_{i}(n) := a_{i} \cdot n  + 1
$$
for $1 \leq i \leq H$ we obtain a set of $H$ admissible linear forms
(here admissibility is trivial since $F_{i}(0) \not \equiv 0 \mod p$
for any prime $p$), hence there exists $i,j$ with $j>i$ such that
$$
|\{ n \leq x : F_{i}(n), F_{j}(n) \text{ both prime} \}|
\gg x/\log^{H} x
$$
Further, given  primes $p = F_{i}(n)$ and $p'= F_{j}(n)$, define
$m= a_{j} \cdot p$ and $m' = a_{i}\cdot p'$.  
Now, since $a_{i} \equiv 0 \mod 4$ for all $i$, $F_{i}(n) \equiv 1 \mod 4$
for all $n$, hence $p,p' \equiv 1 \mod 4$ and consequently $m,m' \in
\mathcal{N}_2$.  
Further,  $m'
\ll_{R} x$; letting $h = m-m'$ we find that
$$
h = m-m' = a_{j} \cdot F_{i}(n) -  a_{i} \cdot F_{j}(n) = a_{j}
-a_{i} 
$$
and thus $0 < h < H^{2}$.
Moreover,
$$
r_{2}(m) = r_{2}(p \cdot a_{j}) = 2 \cdot r_{2}(a_{j})
$$
and similary $r_{2}(m') = 2 \cdot r_{2}(a_{i})$.  Since $r_{2}(a_{j}) \geq
R \cdot r_{2}(a_{i})$  we find that 
$$
r_{2}(m') \geq R \cdot r_{2}(m),
$$
and that $r_{2}(m) = 2 \cdot r_{2}(a_{i}) \ll R^{H}$.  Taking $n = m'$
and $h = m-m'$ we find that the number of $n \ll_{R} x$ with the
desired property is $\gg x/(\log x)^{H}$, thus concluding the proof.
  
\end{proof}
\section{Proof of Theorem \ref{thm:3dscar}}
\label{sec:proof-theor-refthm:3}
We prove Theorem \ref{thm:3dscar} by calculating the Fourier
coefficients of the measure (or more precisely, the coefficents in the
spherical harmonics expansion.) Let $Y(x)$ be a spherical harmonic on
$\mathbb{S}^{2}$. Then by the definition of
$g_{\lambda}=\frac{G_{\lambda}}{\|G_{\lambda}\|}$
(cf. \eqref{eq:def-g-lambda}), and the action of $\Op(Y)$
(cf. \eqref{eq:def-matrix-elem}) we get
\begin{multline}
  \label{eq:3dscar-fourier-coefs}
\langle \Op(Y)g_{\lambda},g_{\lambda}\rangle
=\\=
\frac{\displaystyle\sum_{\xi\in\mathbb{Z}^{3}} 
Y
\left(
\frac{\xi}{|\xi|}
\right)\left(
\frac{1}{|\xi|^{2}-\lambda}
  \right)^{2}}
{\displaystyle\sum_{\xi\in\mathbb{Z}^{3}} 
\left(
\frac{1}{|\xi|^{2}-\lambda}
  \right)^{2}}=
\frac{\displaystyle\sum_{n\in\mathcal{N}_{3}} 
W_{Y}(n)\left(
\frac{1}{n-\lambda}
  \right)^{2}}{\displaystyle\sum_{n\in\mathcal{N}_{3}} 
r_{3}(n)\left(
\frac{1}{n-\lambda}
  \right)^{2}}\\=
\frac{W_{Y}(m)+(m-\lambda)^{2}\displaystyle
  \sum_{n\in\mathcal{N}_{3}\setminus\{m\}} 
W_{Y}(n)\left(
\frac{1}{n-\lambda}
  \right)^{2} }
{r_{3}(m)+(m-\lambda)^{2}\displaystyle\sum_{n\in\mathcal{N}_{3}\setminus\{m\}}  
r_{3}(n)\left( 
\frac{1}{n-\lambda}
  \right)^{2} }
\end{multline}
for any $m\in\mathcal{N}_{3}$.  Now, for $m\in\mathcal{N}_{3}$, define
\[
H_{m}(\lambda):=\sum_{n\in\mathcal{N}_{3}\setminus\{m\}}r_{3}(n)
\left(
\frac{1}{n-\lambda}-\frac{n}{n^{2}+1}
\right)
\]
and  rewrite the ``new'' eigenvalue equation
\eqref{eq:evalue} as
\begin{equation}
  \label{eq:evalue-dim3}
\frac{r_{3}(m)}{m-\lambda}-\frac{m}{m^{2}+1}+H_{m}(\lambda)
=c_{0}\tan
\left(
\frac{\varphi}{2}
\right).
\end{equation}
We can now apply Lemma \ref{lem:nearby-zeros}.  Setting $\lambda = m
+\delta$, let
\[
f_{m}(\delta) :=H_{m}(m+\delta)-\frac{m}{m^{2}+1}-c_{0}\tan
\left(
\frac{\varphi}{2}
\right).
\]
Then 
\begin{equation}
  \label{eq:f-prime-sum}
f_{m}'(\delta)=\sum_{n\in\mathcal{N}_{3}\setminus\{m\}}\frac{r_{3}(n)}{(n-m-\delta)^{2}}
>0
\end{equation}
Notice that for $|\delta|<\frac{1}{2}$ there exists an absolute constant
$C>1$ such that
\begin{equation}
  \label{eq:upper-lower-fprime-bound}
\frac{1}{C}f_{m}'(0)\leq f_{m}'(\delta)\leq Cf_{m}'(0).
\end{equation}
Equation \eqref{eq:evalue-dim3} can now be rewritten as
\[
f_{m}(\delta)=\frac{r_{3}(m)}{\delta}
\]
hence, {\em provided} that we can find $m$ for which the bound
$\sqrt{Cr_{3}(m)/f_{m}'(0)}<\frac{1}{2}$ holds, we may take $I =
[-1/2,1/2]$ in 
Lemma \ref{lem:nearby-zeros} and obtain  an eigenvalue
$\lambda$ such that 
\begin{equation}
  \label{eq:5}
|\lambda-m|<\sqrt{Cr_{3}(m)/f_{m}'(0)}.
\end{equation}
To find $m$ for which the above bound is valid, we proceed as follows.
For $l\in\mathcal{N}_{3}$  fixed, define
\[
\Omega(l):=
\left\{
\frac{\xi}{\|\xi\|}:\|\xi\|^{2}=l,\xi\in\mathbb{Z}^{3}
\right\}
\]
and let $\mathcal{M}_{l}:= \left\{ 4^{k}l:k\in\mathbb{N} \right\}$. 
For $m\in\mathcal{M}_{l}$ we then have $r_{3}(m)=r_{3}(l)$, hence
$r_{3}(m)$ is uniformly bounded; we also note that
$\Omega(m)=\Omega(l)$. Since for any integer $m$ there exists an
integer $m' \not \equiv 0,4,7 \pmod 8$ of bounded distance from $m$,
\eqref{eq:siegel-bound} implies that
\begin{equation}
  \label{eq:3d-derivative-bound}
f_{m}'(0)=\sum_{n\in\mathcal{N}_3}
\frac{r_3(n)}{|n-m|^2}\geq\frac{r_3(m')}{|m'-m|^2}
\gg 
r_3(m')\gg (m')^{1/2-\varepsilon}\gg m^{1/2-\varepsilon}.
\end{equation}
Since $r_{3}(m)$ is uniformly bounded for 
$m\in\mathcal{M}_{l}$, we find that 
$$\sqrt{Cr_{3}(m)/f_{m}'(0)}<m^{-1/4+\epsilon}$$
for all sufficiently large $m \in \mathcal{M}_{l}.$ By the above
argument, we have thus found infinitely many $m$ for which there exist
a nearby new eigenvalue $\lambda$ satisfying
$|m-\lambda|<\sqrt{Cr_{3}(l)/f_{m}'(0)} <m^{-1/4+\epsilon}$. In fact, using
\eqref{eq:3d-derivative-bound} we can apply Lemma
\ref{lem:nearby-zeros} again, to get that 
\eqref{eq:upper-lower-fprime-bound} holds for
$C=1+O(m^{-1/4+\epsilon})$ and $\delta = O(m^{-1/4+\epsilon})$.   Let 
$\Lambda_{l}$ be  
the sequence of these eigenvalues; for $\lambda \in \Lambda_{l}$ we
then find, upon recalling the equality in \eqref{eq:f-prime-sum}, and
that
(\ref{eq:upper-lower-fprime-bound}) is valid since $|m-\lambda| =
O(m^{-1/4+\epsilon})$, that 
\begin{multline}
  \label{eq:bounded-l2-norm}
|m-\lambda|^{2}\sum_{n\in\mathcal{N}_{3}\setminus\{m\}}\frac{r_{3}(n)}{|n-\lambda|^{2}}\leq
\\ 
\frac{(1+o(m^{-1/4}))r_{3}(l)}{f_{m}'(0)}\sum_{n\in\mathcal{N}_{3}\setminus\{m\}}
\frac{r_{3}(n)}{|n-\lambda|^{2}}=(1+O(m^{-1/4+\epsilon}))^{2}r_{3}(l) 
\end{multline}
which is bounded. From now on we restrict $\Lambda_{l}$ to a
subsequence such that the limit
\[
A_{l}:=\lim_{\lambda\in\Lambda_{l}}|m-\lambda|^{2}\sum_{n\in\mathcal{N}_{3}\setminus\{m\}}\frac{r_{3}(n)}{|n-\lambda|^{2}}
\]
exists, and hence by \eqref{eq:bounded-l2-norm} is bounded by $r_{3}(l)$.
Furthermore,  for any spherical
harmonic $Y$, 
\begin{equation}
  \label{eq:vanishing-fourier-coeffs}
|m-\lambda|^{2}\sum_{n\in\mathcal{N}_{3}\setminus\{m\}}
\frac{|W_{Y}(n)|}{|n-\lambda|^{2}}\leq  
Cr_{3}(l)/f_{m}'(0)\sum_{n\in\mathcal{N}_{3}\setminus\{m\}}
\frac{|W_{Y}(n)|}{|n-\lambda|^{2}}. 
\end{equation}
We claim that the RHS converges to 0 as
$m\to\infty$. To see this, write
\[
\sum_{n\in\mathcal{N}_{3}\setminus\{m\}}\frac{|W_{Y}(n)|}{|n-\lambda|^{2}}
\leq
\sum_{\substack{n\in\mathcal{N}_{3}\setminus\{m\}\\n\leq m+
    m^{1/3}}}\frac{|W_{Y}(n)|}{|n-\lambda|^{2}} +
\sum_{\substack{n\in\mathcal{N}_{3}\setminus\{m\}\\n>m+
    m^{1/3}}}\frac{|W_{Y}(n)|}{|n-\lambda|^{2}}.  
\]
For the first sum, using that $|\lambda-n| > 1/2$ for $n \neq m$
together with the bound $W_{Y}(n)\ll m^{1/2-\alpha}$ (using
\eqref{eq:Weyl-sum-cancelation}) we find that for all $n\leq m+m^{1/3}$ in
the summand, the first sum is $\ll m^{1/2-\alpha}$.
For the second sum, the mean value theorem gives  that
\[
W_{Y}(n)\ll
n^{1/2-\alpha}=(n-m)^{1/2-\alpha}+O
\left(
\frac{m}{(n-m)^{1/2+\alpha}}
\right) 
\]
and thus,  
\begin{multline}
  \label{eq:weyl-sum-bound1}
 \sum_{n>m+m^{1/3}}\frac{W_{Y}(n)}{|n-m|^{2}}\ll\\
\sum_{n-m>m^{1/3}}\frac{1}{|n-m|^{3/2+\alpha}} +O
 \left(
\sum_{n-m>m^{1/3}}\frac{m}{|n-m|^{5/2+\alpha}}
 \right)\ll\\
 m^{-1/3(1/2+\alpha)}+m^{1-1/3(3/2+\alpha)}\ll m^{1/2-\alpha}.
\end{multline}
Hence, since $f_{m}'(0)\gg m^{1/2-\varepsilon}$
(cf. \eqref{eq:3d-derivative-bound}), 
\begin{equation}
    \label{eq:weyl-sum-bound2}
      \frac{Cr_{3}(l)}{f'(0)}\sum_{n\in\mathcal{N}_{3}}\frac{W_{Y}(n)}{|n-m|^{2}}\ll
\frac{m^{1/2-\alpha}}{m^{1/2-\varepsilon}}\ll m^{-\alpha+\varepsilon}
  \end{equation}
Thus,  for any fixed spherical harmonic $Y$, and for every
$\lambda\in\Lambda_{l}$,
\[
\langle
\Op(Y)g_{\lambda},g_{\lambda}\rangle=
\begin{cases}
 \frac{W_{Y}(l)}{(1+A_{l}+o(1))r_{3}(l)}+ O(
 \lambda^{-\alpha+\varepsilon}) & \text{if $Y$ is non trivial,}\\
1 & \text{if $Y$ is trivial.}  
\end{cases}.
\]
Since these are the spherical harmonics coefficients of the measure
$\frac{1}{1+A_{l}}\delta_{\Omega(l)}+\frac{A_{l}}{1+A_{l}}\nu$, the
proof is concluded.  (Recall that $\nu$ denotes the uniform measure.)

\section{Proof of Theorem \ref{thm:2dscar}}
\label{sec:proof-2dim}

We start by finding a sequence of new eigenvalues lying close to the
set of old eigenvalues. To do so we will again use Lemma
\ref{lem:nearby-zeros}. Recall that
\begin{equation}
  \label{eq:def-H-F}
F(\lambda)=
\begin{cases}
\text{Constant} & \mbox{(weak coupling)}\\
\sum\limits_{\substack{n\in\mathcal{N}_{2}\\|n-n_{+}(\lambda)|\geq
        n_{+}(\lambda)^{\eta}}}\hspace{-25pt}r_{2}(n)\left(
\frac{1}{n-\lambda}-\frac{n}{n^{2}+1}
\right) & \mbox{(strong coupling)}
\end{cases}
\end{equation}
and in analogy with the three dimensional case we define
\begin{multline}
  \label{def-Hm-dim2}
  H_m(\lambda) =
\sum_{n\in\mathcal{N}_{2}\setminus\{m\}}r_{2}(n)\left(\frac{1}{n-\lambda}-
  \frac{n}{n^2+1}\right)-F(\lambda)=\\
\sum_{n\in I(\lambda)\setminus\{m\}}r_{2}(n)\left(\frac{1}{n-\lambda}-
  \frac{n}{n^2+1}\right) 
\end{multline}
where (for some fixed $\eta > 131/146$)
\[
I(\lambda):=
\begin{cases}
\mathcal{N}_{2} \cap
[n_{+}(\lambda)-n_{+}(\lambda)^{\eta},n_{+}(\lambda)+n_{+}(\lambda)^{\eta}] 
&\mbox{(strong   coupling)}\\
\mathcal{N}_{2}&\mbox{(weak coupling)}
\end{cases}
\]

\begin{prop}
\label{prop:2dim-bdd-norm}
Let $F(\lambda)$ be as above, and given $\gamma\in(0,1/10)$ let
$\mathcal{M}_{\gamma}$ be the set of integers given by Lemma
\ref{lem:almost-prime-elts-in-S}.
%
Then, for any $m\in\mathcal{M}_{\gamma}$, there exists a new eigenvalue
$\lambda$ such that $|\lambda-m| \leq \gamma$ and
  \begin{equation}
    \label{eq:good-stuff}
H_{m}'(\lambda) \leq (1+O(\gamma)) \cdot
\frac{r_{2}(m)}{(m-\lambda)^{2}}
  \end{equation}
as $\gamma \to 0$.
\end{prop}
\begin{proof}
  Given $m \in M_{\gamma}$ we start by  finding (at
  least one) nearby new eigenvalue.  To do so,  rewrite the eigenvalue equation
  (i.e., \eqref{eq:evalue} in the weak coupling limit, or
  \eqref{eq:evalue_strong} in the strong coupling limit) 
$$
\sum_{n \in S} r_{2}(n)
\left(
\frac{1}{n-\lambda} -\frac{n}{n^{2}+1}
\right)
= F(\lambda),
$$
as
$$
r_{2}(m)
\left( \frac{1}{m-\lambda} - \frac{m}{m^{2}+1} \right)
 + H_{m}(\lambda)   = 0
$$
Thus, with $\lambda = m+\delta$, and defining
$f(\delta) = H_{m}(m+\delta)-r_{2}(m)\frac{m}{m^{2}+1}$, we wish
to find (small) 
solutions to
$$
f(\delta) = \frac{r_{2}(m)}{\delta}
$$
Now, by \eqref{eq:def-H-F}, $f'(\delta)$ is always a sum of positive
terms, hence we may drop all terms but one, say the one corresponding
to $k=m+3$ (recall that $m  = n^{2}+1$, hence $k = n^{2}+4$ is a sum of two
squares), and find that 
\begin{equation*}
f'(\delta) 
\geq \frac{r_{2}(m+3)}{((m+3)-(m+\delta))^{2}} 
=\frac{r_{2}(m+3)}{(3-\delta)^{2}} 
\geq \frac{r_{2}(m+3)}{10}
\end{equation*}
for $|\delta| \leq 1/10$.   By Lemma \ref{lem:nearby-zeros}, there exists $\delta_{0}$
such that
$$
f(\delta_{0}) = \frac{r_{2}(m)}{\delta_{0}}
$$
and 
$$
|\delta_{0}| \leq \sqrt{10r_{2}(m)/r_{2}(m+3)} \leq \gamma
$$

Using the above estimate on $\delta$ we next
show that the lower bound on $f'(\delta)$ is essentially given by the
size of
$H'_{m}(m)$.  Since $m_{-} \leq m-1$ and $m_{+} 
\geq m+1$, we find that
$$
\frac{1}{(n-(m+\delta))^{2}} =
\frac{1 + O(\gamma)}{(n-m)^{2}} 
$$
holds for all $n \in \mathcal{N}_{2}\setminus\{m\}$ and $|\delta| \leq
\gamma$.  Thus, 
$$
\min_{|\delta| \leq \gamma} f'(\delta) =\min_{|\delta|\leq\gamma} H_{m}'(m+\delta) = H_{m}'(m)(1+O(\gamma))
$$
for $|\delta| \leq \gamma$.
Hence we may take $A=r_{2}(m)$ and $B =
H_{m}'(m+\delta_{0})(1+O(\gamma))$ in Lemma \ref{lem:nearby-zeros}; on
squaring the estimate $\delta_0 \leq \sqrt{A/B}$ we find that
$$
\delta_{0}^{2} H_{m}'(m+\delta_{0}) \leq (1+O(\gamma)) r_{2}(m)
$$
(for $\gamma$ small.)
In particular,  $\lambda = m + \delta_{0}$ is a new
eigenvalue, and
$$
H_{m}'(\lambda)
\leq
\frac{r_{2}(m)}{(m-\lambda)^{2}}\cdot (1+O(\gamma))
$$
\end{proof}
\begin{rem}
\label{rem:h-generalization}
  The above argument in fact gives the following: if $r_{2}(m+h) \geq
  R \cdot r_{2}(m)>0$ for some $0<h<H$, 
  then (again for $|\delta|< 1/10$),
$$
f'(\delta) \gg \frac{R \cdot r_{2}(m)}{H^{2}}
$$
and thus there exists a nearby new eigenvalue $\lambda = m+\delta_0$
with $|\delta_{0}| \ll H/\sqrt{R}$, and
$$
H_{m}'(\lambda) \ll \frac{ r_{2}(m)}{(m-\lambda)^{2}}.
$$

\end{rem}

For $\gamma\in(0,1/10)$ let $\mathcal{M}_{\gamma}$ be the set given by Lemma
\ref{lem:almost-prime-elts-in-S}, and $\Lambda_{\gamma}$ be the set of
corresponding new eigenvalues given by Proposition
\ref{prop:2dim-bdd-norm}.  By restricting to a subsequence we may 
assume that for any $n\in\mathcal{M}_{\gamma}$, the sets 
\[
\noXi(n):=
\left\{
\frac{\xi}{|\xi|}:|\xi|^{2}=n
\right\}
\]
{\em converge} to a limit set $\noXi(\infty)$ of bounded cardinality.

It is now  straightforward to exhibit scarring in momentum space.
\subsection{Scarring in momentum space}
\label{sec:scars-momentum-part}
In this section we prove the first part of Theorem \ref{thm:2dscar_momentum}.
For any fixed {\em positive} $f\in
C^{\infty}(\mathbb{S}^{1})$ we show that there exists a positive constant $0<c\leq1$ such that
  \begin{equation}
    \label{eq:scar_mmnt_2dim}
   \lim_{\lambda\in\Lambda'} \langle
   \Op(f)g_{\lambda},g_{\lambda}\rangle\geq \frac{c}{|\noXi|}\sum_{\xi\in\noXi}f(\xi) 
  \end{equation}
 Let
$$W_{f}(n):=\displaystyle\sum_{|\xi|^{2}=n}f
\left(
\frac{\xi}{|\xi|}
\right).$$ 
We start with an upper bound on the $L^{2}$ norm of $G_{\lambda}$: By
definition of $G_{\lambda}$, its $L^{2}$ norm is (recall that $\lambda
= m +\delta$  where 
$|\delta| < \gamma$, and $\gamma < 1/10$)
\begin{multline}
  \|G_{\lambda}\|^{2}=\sum_{n\in\mathcal{N}_{2}} \frac{r_{2}(n)}{|n-\lambda|^{2}} = 
\frac{r_{2}(m)}{|m-\lambda|^{2}} + 
\sum_{\substack{n\in\mathcal{N}_{2}\\ n \neq m}}
\frac{r_{2}(n)}{(n-\lambda)^{2}}=\\
\frac{r_{2}(m)}{|m-\lambda|^{2}}+H_{m}'(m+\delta)+\sum_{n\not\in
  I(\lambda)}\frac{r_{2}(n)}{|n-\lambda|^{2}}\leq\\ 
(2+O(\gamma))\frac{r_{2}(m)}{|m-\lambda|^{2}}
 + O(\frac{\lambda^\epsilon}{\lambda^{\eta}})
= 
(2+O(\gamma))\frac{r_{2}(m)}{|m-\lambda|^{2}} 
\end{multline}
where the last inequality follows from Proposition
\ref{prop:2dim-bdd-norm}, and that $r_{2}(n)\ll n^{\varepsilon}$.
Recalling that $f$ is positive, this implies that
\begin{multline}
  \label{eq:matrix-elements-momentum}
  \langle \Op(f)g_{\lambda},g_{\lambda}
  \rangle=\frac{\displaystyle\sum_{n\in\mathcal{N}_2}\frac{W_f(n)}{(n-\lambda)^2}}{\displaystyle\sum_{n\in\mathcal{N}_2}\frac{r_2(n)}{
      (n-\lambda)^2}}\geq\frac{\displaystyle\frac{W_{f}(m)}{
      (m-\lambda)^{2}}}{\displaystyle(2+O(\gamma))\frac{r_{2}(m)}{(m-\lambda)^{2}}}
  =\\
  \frac{1}{2+O(\gamma)} \cdot \frac{W_{f}(m)}{r_{2}(m)}
\to
\frac{1}{2+O(\gamma)} \cdot
\frac{1}{|\noXi(\infty)|}\sum_{\xi\in\noXi(\infty)}f(\xi). 
\end{multline}
By choosing $\gamma$ such that $2+O(\gamma)>0$, Theorem
\ref{thm:2dscar_momentum} is proved.

\begin{rem}
\label{rem:half-mass-singular}
We note that the above construction places mass at least
$1/2+O(\gamma)$ on the singular part.
\end{rem}

\subsection{Scarring in position space}
\label{sec:scars-position-side}
To simplify the notation, we use the following convention throughout
this section:  let $w := (0,2) \in \Z^2$, and for $\lambda$ a fixed
``new'' eigenvalue, and $v \in 
\Z^2$ define
$$
c(v) := c_{\lambda}(v) = \frac{1}{|v|^{2}-\lambda},
\quad \quad C(v,w) := c(v)c(v+w).
$$ 
By the definition of $G_{\lambda}$ and $\Op(a)$ we see that using the
new notation (cf. \eqref{def:greens function},
\eqref{eq:def-quantization-pos})
\[
\Op(e_{w})G_{\lambda}(x,x_{0})=\sum_{v\in\mathbb{Z}^{3}}c(v)e^{iv\cdot
  x_{0}}e^{i(v+w)\cdot x},
\]
and therefore 
\begin{equation}
  \label{eq:position-matrix-elements}
\langle \Op(e_{w}) g_{\lambda}, g_{\lambda}  \rangle=
\frac{ e^{-iw\cdot x_{0}} \cdot   \displaystyle\sum_{v \in \Z^2} 
c(v) c(v+w)
}{\displaystyle\sum_{m \in \mathcal{N}_{2}}
  \frac{r_{2}(m)}{(m-\lambda)^{2}}}
=\frac{\displaystyle e^{-iw\cdot x_{0}} \cdot \sum_{v \in \Z^2}  
C(v,w)
}{\displaystyle\sum_{m \in \mathcal{N}_{2}} \frac{r_{2}(m)}{(m-\lambda)^{2}}}
\end{equation}
As we aim to show that \eqref{eq:position-matrix-elements} is bounded
from below in absolute value, we may assume that $x_{0} = 0$.  We will  show
that the sum in the numerator is essentially bounded from below by two
terms in the sum, namely $v$ such $|v|=|v+w|$.

In what follows, $\gamma\in(0,1/10)$ is small (and to be determined
later), $m = n^{2}+1$ will always denote an element of
$\mathcal{M}_{\gamma}$ (recall that by construction, all elements of
$\mathcal{M}_{\gamma}$ are of this form), and given $n$ we define a vector $u
\in \Z^2$ by
$$
u := (n,-1).
$$

 For $\lambda\in\Lambda_{\gamma}$ let
$m\in\mathcal{M}_{\gamma}$ be the corresponding nearby integer (i.e.,
 $|m-\lambda|<\gamma$ by
Proposition \ref{prop:2dim-bdd-norm}),
set $R=\sqrt{\lambda}$, and let
$$
C_{m} := \{ v \in \R^2 : |v|^{2} = m \}
$$
denote the circle of radius $\sqrt{m}$ centered at the origin. Define
$$
A_{R}  = A_{R,w} := \{ v \in \R^2 : |v| \in [R-|w|,
R+|w|] \}
$$
as the annulus of width $2|w|$ containing $C_{m}$, and let
$$
A_{R}^{*}  = A_{R,w}^{*} := \{ v \in \R^2 : |v| \in [R-|w|,
R+|w|],|v|^{2}\neq m, |v+w|^{2}\neq m \}.
$$

The following Lemma will allow us to bound contribution of the
negative terms in the sum in the  numerator of the right hand side of
(\ref{eq:position-matrix-elements}).
\begin{lem}
  \label{lem:negative-are-far}
If $C(v,w)<0$ for $v\in\mathbb{Z}^{2}$, then
$|v|\in[R-|w|,R+|w|]$. Furthermore, if
we in
addition have $|\langle 
v,w\rangle|\leq 
\frac{\sqrt{R}}{2}$, then $v=(\pm 
n,y)$ with $-3\leq y\leq-1$ provided that $R$ is sufficiently large.
\end{lem}
\begin{proof}
  Since $C(v,w) < 0$ if and only if the line segment joining $v$
  and $v+w$ intersects $C_{R}$, the first assertion follows from the
  triangle inequality. 

 We now write $v=(x,y)$ for $x,y \in \Z$.

{\em First case:}
If $|x|\geq n+1$, then
$$
|v|^{2}-\lambda \geq x^{2} - \lambda \geq (n+1)^{2}
-\lambda = n^{2} +2n+1-\lambda \geq m + 1 - \lambda \geq 1-\delta
$$
and similarly $|v+w|^{2}-\lambda \geq 1-\delta$. Recalling that
$|\delta|<\gamma \leq 1/10$ we find that $C(v,w) >0$.

{\em Second case:}
Assume that $|x| \leq n-1$.  We note that $C(v,w)<0$ implies that
either $|v|^{2} > \lambda$, or that $|v+w|^{2} > \lambda$.

Now, if $|v|^{2}>\lambda$, then 
$$
|v|^{2} = x^{2}+y^{2} > \lambda = m + \delta = n^{2}+1 + \delta
$$
so,
$$
y^{2} > n^{2}+1+\delta - (n-1)^{2} \geq n.
$$
Consequently, $|y| \geq \sqrt{n} > \sqrt{R}/2$, hence $|\langle v,w
\rangle| = 2|y| > \sqrt{R}$ and the claim is vacuous.

On the other hand,  if $|v+w|^2>\lambda$ then, as $x^{2} \leq
(n-1)^{2}$, 
$$
|v+w|^2=x^2+(y+2)^2>\lambda=m+\delta=n^2+1+\delta
$$
so $|y|\geq\sqrt{n}$, and as before $|\langle
v,w\rangle|>\sqrt{R}/2$; again the claim is vacuous.

{\em Third case:}
For $|x| = n$, since $|\delta|=|\lambda-m| < 1/10$ we find that
$$
|v|^{2} -\lambda = n^{2} + y^{2} -\lambda = m -\lambda +
y^{2}-1  = -\delta + y^{2}-1,
$$
and
$$
|v+w|^{2} -\lambda = n^{2} + (y+2)^{2} -\lambda 
= - \delta + (y+2)^{2}-1 
$$ 
so both $c(v),c(v+w)$ are positive for $v = (\pm n,y)$ if  $y
\leq -4$ or $y 
\geq 2$.
\end{proof}
In light of Lemma \ref{lem:negative-are-far}, we consider the
following three sets of points $v\in\mathbb{Z}^{2}$:
\begin{eqnarray*}
&&V_{1}:=\left\{
v\in\mathbb{Z}^{2}: v=(\pm n,y), -3\leq y\leq1
\right\}  \\
&&V_{2}:= \left\{
v\in\mathbb{Z}^{2}: C(v,w)<0, \substack{|v|^2,|v+w|^2\neq  m,\\ 
\sqrt{R}/2\leq|\langle v,w\rangle|\leq 3R} 
  \right\}\\
&&V_{3}:=
\left\{
v\in\mathbb{Z}^{2}:C(v,w)<0,\substack{|v|^2=m\mbox{ or }|v+w|^2=m,\\
  \sqrt{R}/2\leq|\langle v,w\rangle|\leq 3R}
\right\}.
\end{eqnarray*}

Notice that these three sets, for $R$ sufficiently large, cover all
$v\in\mathbb{Z}^{2}$ such that 
$C(v,w)<0$, because if $C(v,w)<0$, then $|\langle v,w\rangle|\leq
(R+|w|)|w|<3R$ for $R>4$. 
\begin{figure}[h]
\ifanswers
\includegraphics[height=6.cm,width=11.5cm]{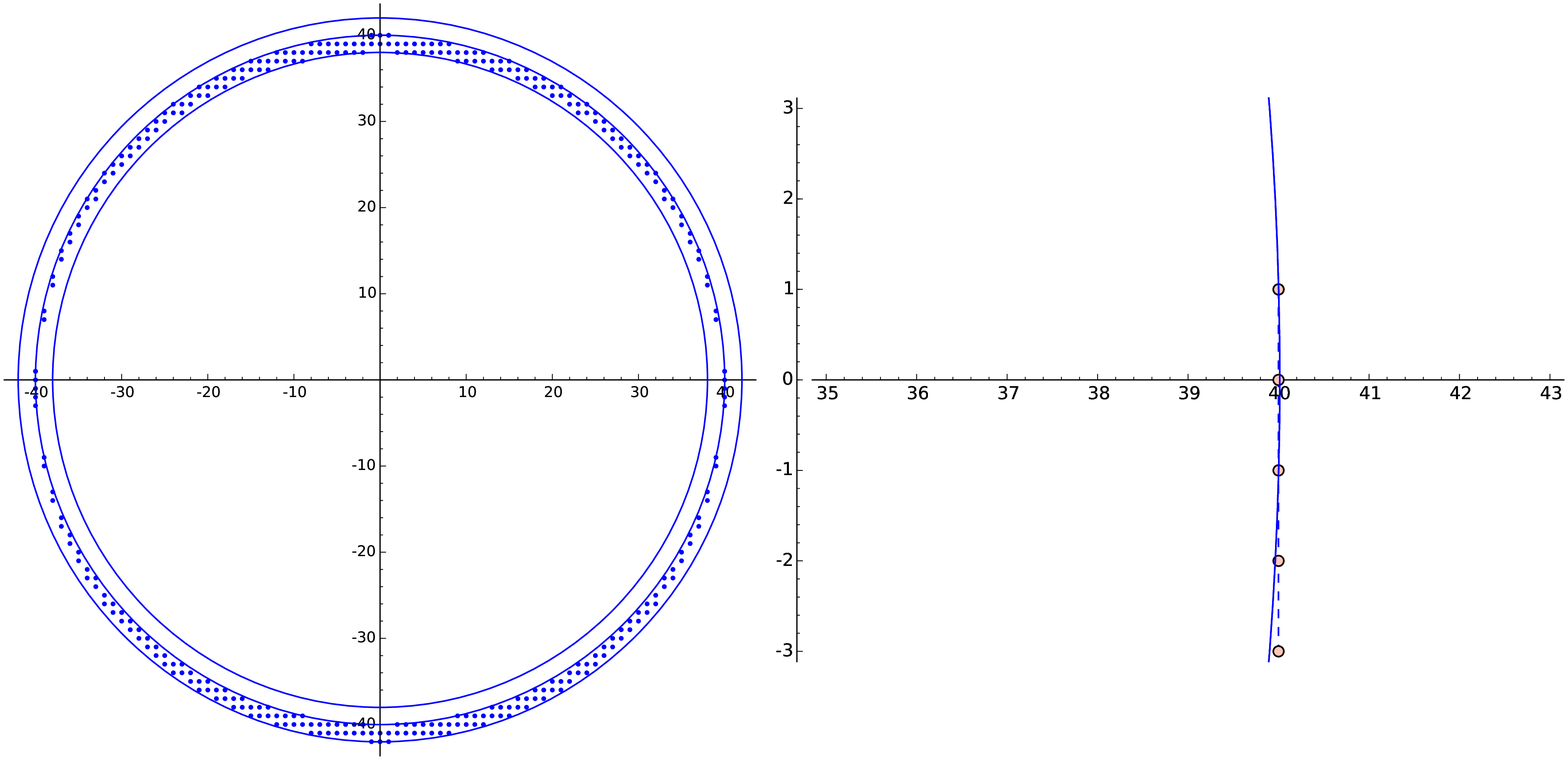}
\fi
\caption{An illustration of setting of Lemma
  \ref{lem:negative-are-far}. For  $m=40^2+1$ and (say)
  $\lambda=0.1$, {\bf only} the lattice points with $C(v,w)<0$ are
  plotted. On the right plot we zoomed around the point (40,0). Notice
  that the only points near $(40,\pm 1)$ lie on the line
  $x=40$.} 
\label{fig:negative-are-far} 
\end{figure}
For ease of notation, we make the following
definitions: For a finite set $X\subset\mathbb{R}$, define
$\operatorname{argmin}_{X}(|x|)$ as the smallest $x \in X$ which
minimizes $|x|$,
and for  $v\in\mathbb{Z}^2$ and $w$ a ``short vector'' as
before, let
\begin{eqnarray*}
\operatorname{Nbr}_w(v)&:=&\{C(v-w,w)+C(v,w),C(v,w)+C(v+w,w)\}\\
 S_w(v)&:=&\mbox{argmin}_{\operatorname{Nbr}_w(v)}(|x|).
\end{eqnarray*}
\begin{cor}
\label{cor:sum-over-negs}
For any $\delta\in (-1/10,1/10)$,  as $R\to\infty$
\begin{equation}
  \label{eq:sum-over-negs1}
    \sum_{v\in\mathbb{Z}^{2}}C(v,w)\geq
  \frac{2}{\delta^{2}}+O(\frac{1}{\delta})+\sum_{v\in V_{2}} S_{w}(v)
\end{equation}
\end{cor}
\begin{proof}
We first notice that if $R$ is large enough, only
one sign change can occur for $|v+tw|^{2}-\lambda$ when $t \in \R$ is bounded,
so if $C(v,w)<0$ then both $C(v-w,w),C(v+w,w)>0$. Also, as mentioned
above, by Lemma \ref{lem:negative-are-far}, if $C(v,w)<0$ then $v$ is
in either $V_{1}, V_{2}$ or
$V_{3}$ for $R$ large. Therefore, after removing only positive terms,
we find that
  \begin{multline}
\sum_{v\in\mathbb{Z}^{2}}C(v,w)\geq \sum_{v\in V_{1}}C(v,w)+\sum_{v\in
  V_{3}}C(v,w)+\sum_{v\in V_{2}}S_{w}(v) =\\2\sum_{y=-3}^{1}C((n,y),w)+\sum_{v\in
  V_{3}}C(v,w)+\sum_{v\in V_{2}}S_{w}(v)
\end{multline} 
Now, by definition of $C(v,w)$, for the first sum we have that
\begin{multline}
  \sum_{y=-3}^{1}C((\pm n,y),w)=\\
\frac{1}{(-\delta)^{2}}+2
\left(
\frac{1}{(3-\delta)(-1-\delta)}-\frac{1}{(8-\delta)\delta}
\right)=\frac{1}{\delta^{2}}+O
\left(
\frac{1}{\delta}
\right).
\end{multline}
(Note that when $\delta \to 0$, the dominant term $1/\delta^{2}$
comes from the term $y=-1$.)

For $v\in V_{3}$ and $|v|^2=m$, we have $|\langle v,w\rangle|\gg \sqrt{R}$,
and so (recall that $m-\lambda = -\delta$)  
\[
||v+w|^{2}-\lambda|=||v|^{2}+2\langle v,w\rangle+|w|^{2}-\lambda|=|2\langle v,w\rangle+|w|^2-\delta|\gg
\sqrt{R}.
\]
and $C(v,w)\ll\frac{1}{\delta \sqrt{R}}$. Using a similar argument we get that $C(v,w)\ll\frac{1}{\delta \sqrt{R}}$ if $|v+w|^2=m$. Therefore, since $r_{2}(m)$ is bounded,
\[
\sum_{v\in V_{3} }C(v,w)\ll\sum_{v\in V_{3}}\frac{1}{\delta
  \sqrt{R}}\ll \frac{1}{\delta \sqrt{R}} 
\]
and \eqref{eq:sum-over-negs1} follows. 
\end{proof}
\begin{figure}[h]
\ifanswers
\begin{pspicture}(0,2)(8.5,8.2)
\psarc[linestyle=dotted](1,5){4}{-35}{35}
\psarc(1,5){5}{-35}{35}
\psarc[linestyle=dotted](1,5){6}{-35}{35}
\psline(4.5,5)(7.5,5)
\rput(4.6,8){$|v|^2=\lambda$}
\rput(5.8,2){\psdot\pnode{A1}}
\rput(5.8,3){\psdot\pnode{A2}}
\rput(5.8,4){\psdot\pnode{A3}}\rput[l](6,4){$u=(n,-1)$}
\rput(5.8,5){\psdot\pnode{A4}}
\rput(5.8,6){\psdot\pnode{A5}}
\rput(5.8,7){\psdot\pnode{A6}}
\rput(5.8,8){\psdot\pnode{A7}}
\ncarc[linestyle=dashed,arcangle=30]{A1}{A3}\pnode(5.5,3){B1}
\ncarc[linestyle=dashed,arcangle=30]{A2}{A4}\pnode(5.5,4){B2}
\ncarc[linestyle=dashed,arcangle=330]{A3}{A5}
\ncarc[linestyle=dashed,arcangle=30]{A4}{A6}\pnode(5.5,6){B4}
\ncarc[linestyle=dashed,arcangle=30]{A5}{A7}\pnode(5.5,7){B5}
\pnode(6,4.7){B3}\cnodeput[linestyle=none](8,4.7){C2}{$C(u,w)=\frac{1}{\delta^2}$}
\rput(2,4.7){\rnode{C1}{$\sum C(v,w)=O(\frac{1}{\delta})$}}
\ncarc{B3}{C2}
\ncarc{B1}{C1}\ncarc{B2}{C1}\ncarc[arcangle=330]{B4}{C1}\ncarc[arcangle=330]{B5}{C1}
\end{pspicture}
\fi
\label{fig:pos-main-contr}
\caption{The main contribution in Corollary \ref{cor:sum-over-negs} is
seen here. $C(u,w)=c(u)c(u+w)=\frac{1}{\delta^{2}}$, and all other
points $v$ with $C(v,w)<0$ contribute $O(\frac{1}{\delta})$.}
\end{figure}
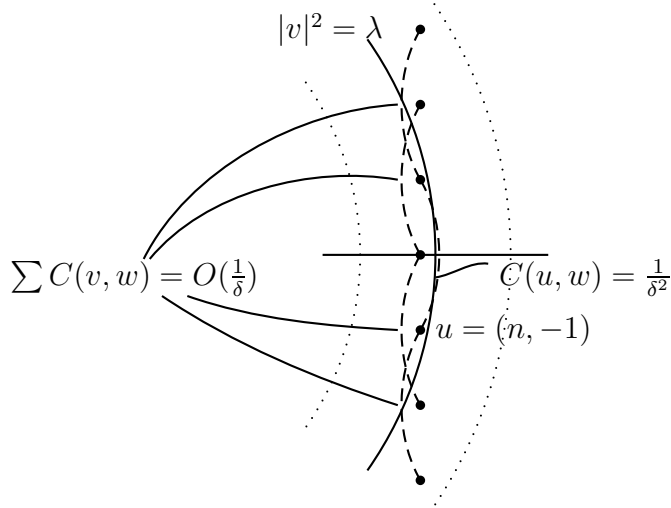

The contribution from the remaining, more subtle, terms are treated
by  ``pairing off'' negative summands with positive ones, and thereby
getting some extra savings.
\begin{lem}
  \label{lem:nearby-terms-cancel}
  Let $v \in V_{2}$ be an element such that $|\langle v,  w \rangle|
  \geq R^{1/3}$. Then
$S_{v}(w)\ll\frac{\log^{2} |\langle v, w \rangle|}{|\langle v, w
  \rangle|^{2}}$, as $R \to \infty$.
\end{lem}
\begin{proof}
For notational convenience, we put
$B := B_{v} = |\langle v, w \rangle|$.
We split the proof into two cases.  

{\em First case:} Here we assume that $||v|^{2}-\lambda| \leq
||v+w|^{2}-\lambda|$. 
If $||v|^{2}-\lambda| \geq B/\log B$, then
$$
|c(v) c(v+w)| \leq \frac{\log^{2}B }{ B^{2}}.
$$
We may therefore assume that $||v|^{2}-\lambda| \leq B/ \log B.$
Now, 
$$
|v+w|^{2}-\lambda = 
|v|^{2} + 2 \langle v, w \rangle +|w|^{2} - \lambda
= 
|v|^{2} -\lambda \pm 2 B +|w|^{2}
$$
and similarly $|v-w|^{2}-\lambda = |v|^{2} -\lambda - (\pm 2 B)
+|w|^{2}$, hence
\begin{multline}
\label{eq:two-term-bound}
C(v,w) + C(v-w,w) = c(v) \cdot c(v+w) + c(v-w) \cdot c(v) =
\\=
\frac{1}{|v|^{2}-\lambda}
\left(
\frac{2(|v|^{2} + |w|^{2}-\lambda)  }
{(|v|^{2} + 2B+ |w|^{2}-\lambda) (|v|^{2} - 2B+ |w|^{2}-\lambda)}
\right).
\end{multline}
(note that the two $\pm 2B$ terms above occur with {\em opposite} signs).
Recalling the assumption  $|v|^{2} \neq m$, together with
$|w|^{2} =4$, we find that
$||v|^{2}-\lambda| \geq 1/2$ (note that $|\lambda-m| \leq\delta\leq
1/10$ by our assumption on $\delta$), and this together with 
\eqref{eq:two-term-bound} shows that 
$$
C(v,w) + C(v-w,w)
\ll
\frac{1}
{(|v|^{2} + 2B+ |w|^{2}-\lambda) (|v|^{2} - 2B+ |w|^{2}-\lambda)}.
$$
Since we assume that $||v|^{2}-\lambda| \leq B/\log B$, we find that
$$
| |v|^{2} \pm 2B+ |w|^{2}-\lambda| \gg B
$$
and thus $C(v,w) + C(v-w,w)\ll 1/B^{2}$.

{\em Second case.}  Here we assume that $||v|^{2}-\lambda| >
||v+w|^{2}-\lambda| $.  This case follows by a similar argument,
except for showing that
$$
|c(v+w)( c(v) + c(v+2w))| \ll  \frac{\log^{2}B}{B^{2}}.
$$

\end{proof}  
\begin{figure}
\ifanswers
\begin{pspicture}(3,5)(10,8)
\psset{unit=1.2}
\psarc[linestyle=dotted](5,0){4}{65}{115}
\psarc(5,0){5}{65}{115}
\rput(2.3,4.3){$|v|^2=\lambda$}
\psarc[linestyle=dotted](5,0){6}{65}{115}
\dotnode(5.6,4.85){A1}\rput[l](5.8,4.8){$v$}
\dotnode(5.6,4.45){A2}\rput[l](5.8,4.45){$v-w$}
\dotnode(5.6,5.25){A3}\rput[l](5.8,5.25){$v+w$}
\psline(5.6,4.45)(5.6,5.25)
\dotnode(4.5,5.05){A1}\rput[r](4.3,5){$v+w$}
\dotnode(4.5,4.65){A2}\rput[r](4.3,4.6){$v$}
\dotnode(4.5,5.45){A3}\rput[r](4.3,5.4){$v+2w$}
\psline(4.5,4.65)(4.5,5.45)
\pnode(4.51,4.85){B1}
\pnode(4.51,5.25){B2}
\pnode(5.59,4.65){C1}
\pnode(5.59,4.95){C2}
\rput(6,4){$C(v,w)+C(v-w,w)\ll\frac{\log^2B}{B^2}$}
\pnode(6,4.15){C0}\pnode(4.5,4.15){C3}
\ncarc{->}{C1}{C0}\ncarc[arcangle=350]{->}{C2}{C3}
\rput[l](3.1,6.5){\rnode{B0}{$C(v+w,w)+C(v,w)\ll\frac{\log^2B}{B^2}$}}
\pnode(6,6.25){B0}\pnode(4.58,6.20){B3}
\ncarc{->}{B1}{B0}\ncarc[arcangle=300]{->}{B2}{B3}
\end{pspicture}
\fi
\label{fig:pos-scar-negl-cont}
\vspace{.5cm}
\caption{An illustration of Lemma \ref{lem:nearby-terms-cancel}. For
  $v\in V_{2}$, at least one of the terms $C(v,w)+C(v+w,w)$ or
  $C(v,w)+C(v-w,w)$is $\ll\frac{\log^{2}B}{B^{2}}$ }
\end{figure}
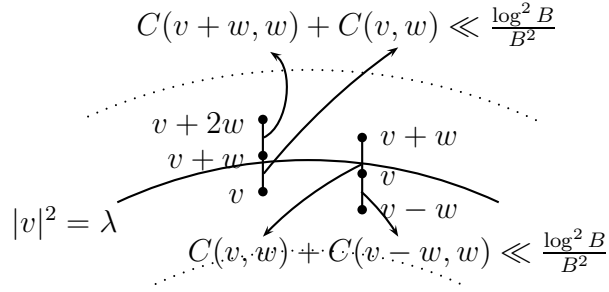
\begin{cor}
\label{cor:small-neg-contr}
As $R\to\infty$,  we have
$$\displaystyle \sum_{v\in V_{2}}S_{w}(v)= \sum_{\substack{v\in A_{R}^{*}\\
      C(v,w)<0\\| \langle v,w\rangle|\in[R^{1/3},3R]}}
  S_{w}(v)=o(1).
$$
\end{cor}
\begin{proof}
Write
\begin{multline}
\sum_{\substack{v\in A_{R}^{*}\\
      C(v,w)<0\\| \langle v,w\rangle|\in[R^{1/3},3R]}}
  S_{w}(v)=\\
\sum_{\substack{k\in\mathbb{N}\\R^{1/3}\leq 2^{k}<R/\log R}}\sum_{\substack{v\in A_{R}^{*}\\
      C(v,w)<0\\| \langle v,w\rangle|\in[2^{k},2^{k+1})}}S_{w}(v)+\sum_{\substack{v\in A_{R}^{*}\\
      C(v,w)<0\\| \langle v,w\rangle|\in[R/\log R,3R]}}S_{w}(v)
\end{multline}
Since the number of lattice points in $A_{R}$ satisfying $\langle
v,w\rangle\in I$ is $O(|I||w|)$ for any interval $I \subset [-R/\log R,
R/\log R]$, we get by Lemma
\ref{lem:nearby-terms-cancel} that
\[
\sum_{\substack{k\in\mathbb{N}\\R^{1/3}\leq 2^{k}<R/\log R}}\sum_{\substack{v\in A_{R}^{*}\\
      C(v,w)<0\\| \langle
      v,w\rangle|\in[2^{k},2^{k+1})}}S_{w}(v)\ll\sum_{2^{k}\geq
    R^{1/3}}2^{k}
  \left(
\frac{k}{2^{k}}
  \right)^{2}\ll\frac{1}{R^{1/3-\varepsilon}}=o(1) 
\]
and
\[
\sum_{\substack{v\in A_{R}^{*}\\
      C(v,w)<0\\| \langle v,w\rangle|\in[R/\log R,3R]}}S_{w}(v)\ll
 R \frac{ \log^{2}R}{(R/\log R)^{2}} =  
\frac{ \log^{4}R}{R} = o(1)
\]
\end{proof}
\subsubsection{Conclusion}
\label{sec:conclusion}
We can now conclude the proof of Theorem \ref{thm:2dscar_momentum} by
proving that for $f(x)=e^{\pi i \langle x, w \rangle}\in C^\infty(\mathbb{T}^{2})$,
  \begin{equation}
    \label{eq:scar_pos_2dim}
    \lim_{\lambda\in\Lambda'}\langle \Op(f)g_{\lambda},g_{\lambda}\rangle>0
  \end{equation}
 By \eqref{eq:position-matrix-elements} (recall that we 
may assume that
 $x_0 =0$),
\[
\langle \\Op(e_{w}) g_{\lambda}, g_{\lambda}  \rangle=
\frac{\displaystyle\sum_{v \in \Z^2} C(v,w)
}{\displaystyle\sum_{m \in \mathcal{N}_{2}} \frac{r_{2}(m)}{(m-\lambda)^{2}}}
\]
By corollaries \ref{cor:sum-over-negs} and \ref{cor:small-neg-contr}
\begin{equation}
  \sum_{v \in \Z^2} 
C(v,w)\geq
\frac{2}{\delta^{2}}+O(\frac{1}{\delta})
\end{equation}
On the other hand, by Proposition \ref{prop:2dim-bdd-norm}
\[
\displaystyle\sum_{n \in \mathcal{N}_{2}} \frac{r_{2}(n)}{(n-\lambda)^{2}}=\frac{r_{2}(m)}{(m-\lambda)^{2}}+H_{m}'(\lambda)\leq(2+O(\gamma))\frac{r_{2}(m)}{(m-\lambda)^{2}}
\]
and, recalling that  $r_{2}(m)$ is bounded, we can choose
$\gamma$ small enough  (recall that $|m-\lambda|=|\delta|\leq\gamma$)  so  that
\[
\displaystyle{\langle \Op(e_{w}) g_{\lambda}, g_{\lambda}
\rangle\geq\frac{\displaystyle\frac{2}{\delta^{2}}+
  O(\frac{1}{\delta})}
{\displaystyle(2+O(\gamma))\frac{r_{2}(m)}{\delta^{2}}}}
=
\frac{2+O(\delta)}{(2+O(\gamma))r_{2}(m)} 
=
\frac{2+O(\gamma)}{(2+O(\gamma))r_{2}(m)} 
\]
is {\em uniformly} bounded from below.

\section{Proof of
  Theorem~\ref{thm:large-density-scars}}
\label{sec:version-2}
Recall first the setting proved in Proposition \ref{prop:r2-boost}:
There exist an integer $H \geq 1$ with the property that for
all sufficiently large $R$ there exist an integer $h  \in (0,H^{2})$ such that
\begin{multline*}
|\{ n \in \mathcal{N}_2 : n \leq x, \quad 0 < r_{2}(n) \ll R^{H},
\quad r_{2}(n+h) \geq R \cdot r_{2}(n)  \}|
\\
\gg_{R} x/(\log x)^{H}
\end{multline*}
as $x \to \infty$.

As noted in Remark~\ref{rem:h-generalization}, if $r_{2}(m+h) \geq R
\cdot r_{2}(m) > 0$ for some integer $h$ such that $0<h<H$, then there
exists a new eigenvalue $\lambda = m + \delta_{0}$ with
$$
\delta_{0} \ll 2H/\sqrt{R}, 
\quad \quad
H_{m}'(\lambda) \ll \frac{ r_{2}(m)}{(m-\lambda)^{2}}.
$$
The argument in Section~\ref{sec:scars-momentum-part} then shows that
$\lambda$ gives rise to a momentum scar provided $r_2(m)$ is also
bounded.  Proposition~\ref{prop:r2-boost} then gives, upon choosing
$R$ sufficently large, that the number of such $m \leq x$ is $\gg
x/(\log x)^{H}$.

\bibliography{scarbib}
\bibliographystyle{abbrv}

\end{document}